\documentclass{llncs}

\usepackage[utf8]{inputenc}

\usepackage{etoolbox}
    \AtEndEnvironment{proof}{\null\hfill\ensuremath\square}
\usepackage{float}
\usepackage{microtype}
\usepackage{xcolor}
    \definecolor{medium-blue}{rgb}{0,0,0.5}
\usepackage{algorithm}
\usepackage{algpseudocode}

\usepackage{hyperref}
\usepackage{mathtools}
    \mathtoolsset{centercolon}
\usepackage[nameinlink]{cleveref}
\usepackage{amsfonts}
\usepackage{tikz}
    \usetikzlibrary{calc}
    \usetikzlibrary{arrows}
    \usetikzlibrary{decorations.pathreplacing}

\usepackage{graphicx}
    \graphicspath{{images/}}
\usepackage{nicefrac}
\usepackage{xifthen}

\newcommand{\C}{\mathbb{C}}

\newcommand{\s}{\sqrt{2}}
\newcommand{\wlofg}{without loss of generality}
\DeclareMathOperator{\mysum}{sum}
\DeclareMathOperator*{\argmin}{arg\,min}

\newcommand{\entry}[7]{
\begin{tikzpicture}[scale=#1,baseline={([yshift={-\ht\strutbox}]current bounding box.north)}]\input{tikz/#2.tex}\end{tikzpicture}
& \textbf{#3} (#4) \vfill\vspace{10pt}
\ifthenelse{\isempty{#5}}{}{\textit{Condition:} #5 \vfill\vspace{5pt}}
\textit{Density:} #6 \vfill\vspace{5pt}
\textit{Approximation factor:} #7\\
}

\makeatletter
    \g@addto@macro\@floatboxreset\centering

    \def\fps@figure{H}
    \def\fps@table{H}
    \def\fps@algorithm{H}
\makeatother

\newcommand{\B}{1.36207}

\tikzset{tight/.style={inner sep=1pt}}
\tikzset{bracket/.style={inner sep=10pt,draw=gray,decorate,decoration={brace,amplitude=5pt}}}
\tikzset{helper/.style={dashed}}
\tikzset{filled/.style={fill=black!40!white}}
\tikzset{filled2/.style={fill=black!15!white}}
\tikzset{cover/.style={fill=black!30!white,fill opacity=0.5}}

\newcommand\hatshape[5][]{
    \pgfmathsetmacro{\r}{sqrt((#2)/pi)}
    \pgfmathsetmacro{\s}{sqrt((#3)/pi)}
    \pgfmathsetmacro{\a}{(#4)}
    \pgfmathsetmacro{\b}{(#5)}

    \coordinate (top) at ({cos((90+(\a-\b)/2))*(\r/cos((\a+\b)/2))},{sin((90+(\a-\b)/2))*(\r/cos((\a+\b)/2)});
    \coordinate (left) at ({-\r/tan(\a/2)},-\r);
    \coordinate (right) at ({\r/tan(\b/2)},-\r);

    \coordinate (leftcenter) at ($(left)+({\s/tan(\a/2)},\s)$);
    \coordinate (leftbottom) at ($(leftcenter)+(-90:\s)$);
    \coordinate (lefttop) at ($(leftcenter)+({-270+\a}:\s)$);

    \coordinate (rightcenter) at ($(right)+({-\s/tan(\b/2)},\s)$);
    \coordinate (rightbottom) at ($(rightcenter)+(-90:\s)$);
    \coordinate (righttop) at ($(rightcenter)+({90-\b}:\s)$);

    \coordinate (topcenter) at ($(top)-({cos((90+(\a-\b)/2))*(\s/cos((\a+\b)/2))},{sin((90+(\a-\b)/2))*(\s/cos((\a+\b)/2)})$);
    \coordinate (topleft) at ($(topcenter)+({-270+\a}:\s)$);
    \coordinate (topright) at ($(topcenter)+({90-\b}:\s)$);

    \coordinate (midcenter) at (0,0);

    \draw[#1] (rightbottom) arc (-90:{90-\b}:\s) -- (topright) arc ({90-\b}:{90+\a}:\s) -- (lefttop) arc ({-270+\a}:-90:\s) -- cycle;
}
\newcommand\gemshape[3][]{
    \pgfmathsetmacro{\r}{sqrt((#2)/pi)}
    \pgfmathsetmacro{\s}{sqrt((#3)/pi)}
    \pgfmathsetmacro{\l}{0.85955*sqrt(#3)}
    \pgfmathsetmacro{\ll}{0.7654*\l}
    \pgfmathsetmacro{\a}{45}
    \pgfmathsetmacro{\b}{45}

    \coordinate (top) at ({cos((90+(\a-\b)/2))*(\r/cos((\a+\b)/2))},{sin((90+(\a-\b)/2))*(\r/cos((\a+\b)/2)});
    \coordinate (left) at ({-\r/tan(\a/2)},-\r);
    \coordinate (right) at ({\r/tan(\b/2)},-\r);

    \coordinate (leftcenter) at ($(left)+({\s/tan(\a/2)},\s)$);
    \coordinate (lefttop) at ($(left)+(45:\l)$);
    \coordinate (leftmid) at ($(left)+(22.5:\ll)$);
    \coordinate (leftbottom) at ($(left)+(0:\l)$);

    \coordinate (rightcenter) at ($(right)+({-\s/tan(\b/2)},\s)$);
    \coordinate (righttop) at ($(right)+(135:\l)$);
    \coordinate (rightmid) at ($(right)+(157.5:\ll)$);
    \coordinate (rightbottom) at ($(right)+(180:\l)$);

    \coordinate (midcenter) at (0,0);

    \draw[#1] (top) -- (righttop) -- (rightmid) -- (rightbottom) -- (leftbottom) -- (leftmid) -- (lefttop) -- cycle;
}
\newcommand\sharpgemshape[3][]{
    \pgfmathsetmacro{\r}{sqrt((#2)/pi)}
    \pgfmathsetmacro{\s}{sqrt((#3)/pi)}
    \pgfmathsetmacro{\l}{0.85955*sqrt(#3)}
    \pgfmathsetmacro{\ll}{0.7654*\l}
    \pgfmathsetmacro{\a}{45}
    \pgfmathsetmacro{\b}{45}

    \coordinate (top) at ({cos((90+(\a-\b)/2))*(\r/cos((\a+\b)/2))},{sin((90+(\a-\b)/2))*(\r/cos((\a+\b)/2)});
    \coordinate (left) at ({-\r/tan(\a/2)},-\r);
    \coordinate (right) at ({\r/tan(\b/2)},-\r);

    \coordinate (rightcenter) at ($(right)+({-\s/tan(\b/2)},\s)$);
    \coordinate (righttop) at ($(right)+(135:\l)$);
    \coordinate (rightmid) at ($(right)+(157.5:\ll)$);
    \coordinate (rightbottom) at ($(right)+(180:\l)$);

    \coordinate (midcenter) at (0,0);

    \draw[#1] (top) -- (righttop) -- (rightmid) -- (rightbottom) -- (left) -- cycle;
}

\newcommand\hatsinsquare[1]{
    \draw (0,0) rectangle (\B,\B);

    \pgfmathparse{\B-sqrt((1-(#1))/pi)}
    \begin{scope}[shift={(\pgfmathresult,\pgfmathresult)}]
        \begin{scope}[rotate=-45]
            \pgfmathsetmacro{\hata}{1-(#1)}
            \pgfmathsetmacro{\hatb}{1-2*(#1)}
            \hatshape[filled2]{\hata}{\hatb}{45}{45}
            \node at (midcenter) {\hata};
        \end{scope}
    \end{scope}

    \def\comparg{#1}
    \if\comparg0\else
        \pgfmathparse{(sqrt((#1)/pi)}
        \begin{scope}[shift={(\pgfmathresult,\pgfmathresult)}]
            \begin{scope}[rotate=-225]
                \hatshape[filled2]{#1}{0}{45}{45}
                \node at (midcenter) {#1};
            \end{scope}
        \end{scope}
    \fi
}

\newcommand\hatsinrect[1]{
    \draw (0,0) rectangle (1.5607,1);

    \pgfmathsetmacro{\hatx}{1-(#1)}
    \pgfmathsetmacro{\hata}{(1-(#1))*0.7853}
    \pgfmathsetmacro{\hatb}{(1-2*(#1))*0.7853}

    \pgfmathparse{sqrt(\hata/pi)}
    \begin{scope}[shift={(1.5607-\pgfmathresult,1-\pgfmathresult)}]
        \begin{scope}[rotate=-32.65]
            \hatshape[filled2]{\hata}{\hatb}{32.65}{57.35}
            \node at (midcenter) {$\hatx a$};
        \end{scope}
    \end{scope}

    \def\comparg{#1}
    \if\comparg0\else
        \pgfmathparse{(sqrt((#1)*0.7853/pi)}
        \begin{scope}[shift={(\pgfmathresult,\pgfmathresult)}]
            \begin{scope}[rotate=-212.65]
                \hatshape[filled2]{#1*0.7853}{0}{32.65}{57.35}
                \node at (midcenter) {$#1 a$};
            \end{scope}
        \end{scope}
    \fi
}

\newcommand\gemsinsquare[1]{
    \draw (0,0) rectangle (\B,\B);

    \pgfmathparse{\B-sqrt((1-(#1))/pi)}
    \begin{scope}[shift={(\pgfmathresult,\pgfmathresult)}]
        \begin{scope}[rotate=-45]
            \pgfmathsetmacro{\hata}{1-(#1)}
            \pgfmathsetmacro{\hatb}{1-2*(#1)}
            \gemshape[filled2]{\hata}{\hatb}
            \node at (midcenter) {\hata};
        \end{scope}
    \end{scope}

    \def\comparg{#1}
    \if\comparg0\else
        \pgfmathparse{(sqrt((#1)/pi)}
        \begin{scope}[shift={(\pgfmathresult,\pgfmathresult)}]
            \begin{scope}[rotate=-225]
                \gemshape[filled2]{#1}{0}
                \node at (midcenter) {#1};
            \end{scope}
        \end{scope}
    \fi
}

\newcommand\hatsinhat[5][1]{
    \pgfmathsetmacro{\area}{(#1)}
    \pgfmathsetmacro{\a}{(#2)}
    \pgfmathsetmacro{\b}{(#3)}
    \pgfmathsetmacro{\x}{(#4)}
    \pgfmathsetmacro{\round}{(#5)}
    \pgfmathsetmacro{\f}{(cos(\b/2)^2*sec(\a/2+\b/2)^2*(1-sin(\a)))}
    \pgfmathsetmacro{\g}{(cos(\a/2)^2*sec(\a/2+\b/2)^2*(1-sin(\b)))}

    \hatshape{\area}{\round}{\a}{\b}

    \def\comparg{\x}
    \if\comparg0\else
        \pgfmathparse{sqrt(\area/pi)/tan(\b/2)-sqrt(((\x))/pi)/tan(\b/2)}
        \begin{scope}[shift={(\pgfmathresult,0)}]
            \pgfmathparse{sqrt(\area/pi)-sqrt(\x/pi)}
            \begin{scope}[shift={(0,-\pgfmathresult)}]
                \pgfmathsetmacro{\hata}{\x}
                \pgfmathsetmacro{\hatb}{max(\round,\x-\g*(1-\x)/(\f)))}
                \hatshape[filled2]{\hata}{\hatb}{90}{\b}
                \node at (midcenter) {$\hata a$};
            \end{scope}
        \end{scope}
    \fi

    \def\comparg{\x}
    \if\comparg1\else
        \pgfmathparse{sqrt(\area/pi)/tan(\a/2)-sqrt(((1-\x))/pi)/tan(\a/2)}
        \begin{scope}[shift={(-\pgfmathresult,0)}]
            \pgfmathparse{sqrt(\area/pi)-sqrt((1-\x)/pi)}
            \begin{scope}[shift={(0,-\pgfmathresult)}]
                \pgfmathsetmacro{\hata}{1-\x}
                \pgfmathsetmacro{\hatb}{\round}
                \pgfmathsetmacro{\hatb}{max(\round,(1-\x)-\f*(\x)/(\g)))}
                \hatshape[filled2]{\hata}{\hatb}{\a}{90}
                \node at (midcenter) {$\hata a$};
            \end{scope}
        \end{scope}
    \fi
}

\newcommand\hatsinhatnotext[5][1]{
    \pgfmathsetmacro{\area}{(#1)}
    \pgfmathsetmacro{\a}{(#2)}
    \pgfmathsetmacro{\b}{(#3)}
    \pgfmathsetmacro{\x}{(#4)}
    \pgfmathsetmacro{\round}{(#5)}
    \pgfmathsetmacro{\f}{(cos(\b/2)^2*sec(\a/2+\b/2)^2*(1-sin(\a)))}
    \pgfmathsetmacro{\g}{(cos(\a/2)^2*sec(\a/2+\b/2)^2*(1-sin(\b)))}

    \hatshape{\area}{\round}{\a}{\b}

    \def\comparg{\x}
    \if\comparg0\else
        \pgfmathparse{sqrt(\area/pi)/tan(\b/2)-sqrt(((\x))/pi)/tan(\b/2)}
        \begin{scope}[shift={(\pgfmathresult,0)}]
            \pgfmathparse{sqrt(\area/pi)-sqrt(\x/pi)}
            \begin{scope}[shift={(0,-\pgfmathresult)}]
                \pgfmathsetmacro{\hata}{\x}
                \pgfmathsetmacro{\hatb}{max(\round,\x-\g*(1-\x)/(\f)))}
                \hatshape[filled2]{\hata}{\hatb}{90}{\b}
            \end{scope}
        \end{scope}
    \fi

    \def\comparg{\x}
    \if\comparg1\else
        \pgfmathparse{sqrt(\area/pi)/tan(\a/2)-sqrt(((1-\x))/pi)/tan(\a/2)}
        \begin{scope}[shift={(-\pgfmathresult,0)}]
            \pgfmathparse{sqrt(\area/pi)-sqrt((1-\x)/pi)}
            \begin{scope}[shift={(0,-\pgfmathresult)}]
                \pgfmathsetmacro{\hata}{1-\x}
                \pgfmathsetmacro{\hatb}{\round}
                \pgfmathsetmacro{\hatb}{max(\round,(1-\x)-\f*(\x)/(\g)))}
                \hatshape[filled2]{\hata}{\hatb}{\a}{90}
            \end{scope}
        \end{scope}
    \fi
}

\newcommand\gemsingem[2]{
    \pgfmathsetmacro{\a}{45}
    \pgfmathsetmacro{\b}{45}
    \pgfmathsetmacro{\x}{(#1)}
    \pgfmathsetmacro{\round}{(#2)}
    \pgfmathsetmacro{\f}{(cos(\b/2)^2*sec(\a/2+\b/2)^2*(1-sin(\a)))}
    \pgfmathsetmacro{\g}{(cos(\a/2)^2*sec(\a/2+\b/2)^2*(1-sin(\b)))}

    \gemshape{1}{\round}

    \def\comparg{\x}
    \if\comparg0\else
        \pgfmathparse{sqrt(1/pi)/tan(\b/2)-sqrt(((\x))/pi)/tan(\b/2)}
        \begin{scope}[shift={(\pgfmathresult,0)}]
            \pgfmathparse{sqrt(1/pi)-sqrt(\x/pi)}
            \begin{scope}[shift={(0,-\pgfmathresult)},rotate=135,xscale=-1]
                \pgfmathsetmacro{\hata}{\x}
                \pgfmathsetmacro{\hatb}{max(\round,\x-\g*(1-\x)/(\f)))}
                \sharpgemshape[filled2]{\hata}{\hatb}
                \node at (midcenter) {$\hata$};
            \end{scope}
        \end{scope}
    \fi

    \def\comparg{\x}
    \if\comparg1\else
        \pgfmathparse{sqrt(1/pi)/tan(\a/2)-sqrt(((1-\x))/pi)/tan(\a/2)}
        \begin{scope}[shift={(-\pgfmathresult,0)}]
            \pgfmathparse{sqrt(1/pi)-sqrt((1-\x)/pi)}
            \begin{scope}[shift={(0,-\pgfmathresult)}]
                \pgfmathparse{\x < 0.1715 ? 0 : -135}
                \begin{scope}[rotate=\pgfmathresult]
                    \pgfmathparse{\x < 0.1715 ? 1 : -1}
                    \begin{scope}[xscale=\pgfmathresult]
                        \pgfmathsetmacro{\hata}{1-\x}
                        \pgfmathsetmacro{\hatb}{\round}
                        \pgfmathsetmacro{\hatb}{\x < 1/3 ? (1-\x) : max(\round,(1-\x)-\f*(\x)/(\g))}
                        \pgfmathsetmacro{\fillstyle}{\x < 1/3 ? "filled" : "filled2"}
                        \sharpgemshape[\fillstyle]{\hata}{\hatb}
                        \node at (midcenter) {$\hata$};
                    \end{scope}
                \end{scope}
            \end{scope}
        \end{scope}
    \fi
}

\usepackage[english]{babel}

\hypersetup{
    colorlinks, linkcolor={medium-blue},
    citecolor={medium-blue}, urlcolor={medium-blue}
}

\title{Split Packing: Algorithms for Packing Circles with Optimal Worst-Case Density\thanks{Extended abstracts presenting parts of this paper appeared in the
27th ACM-SIAM Symposium on Discrete Algorithms (SODA 2017)~\cite{morr2017split} and the 15th Algorithms and Data Structures Symposium (WADS 2017)~\cite{wads2017}.}}

\author{Sándor P. Fekete \and Sebastian Morr \and Christian Scheffer}
\institute{Department of Computer Science, TU Braunschweig, Germany\\\email{s.fekete@tu-bs.de}, \email{sebastian@morr.cc}, \email{scheffer@ibr.cs.tu-bs.de}}

\begin{document}

\maketitle

\begin{abstract}
    In the classic \emph{circle packing problem}, one asks whether a given set of circles can be packed into a given container.
    Packing problems like this have been shown to be $\mathsf{NP}$-hard.
    In this paper, we present new sufficient conditions for packing circles into square and triangular containers, using only the sum of the circles' areas:
    For square containers, it is possible to pack any set of circles with a combined area of up to $\approx \! 53.90\%$ of the square's area.
    And when the container is a right or obtuse triangle, any set of circles whose combined area does not exceed the triangle's incircle can be packed.
    
    These area conditions are tight, in the sense that for any larger areas, there are sets of circles which cannot be packed.
    Similar results have long been known for squares, but to the best of our knowledge, we give the first results of this type for circular objects.

    Our proofs are constructive:
    We describe a versatile, divide-and-conquer-based algorithm for packing circles into various container shapes with optimal worst-case density.
    It employs an elegant subdivision scheme that recursively splits the circles into two groups and then packs these into subcontainers.
    We call this algorithm \emph{Split~Packing}.
    It can be used as a constant-factor approximation algorithm when looking for the smallest container in which a given set of circles can be packed, due to its polynomial runtime.


    A browser-based, interactive visualization of the Split Packing approach and other related material can be found at \url{https://morr.cc/split-packing/}.
\end{abstract}

\section{Introduction}

Given a set of circles, can you decide whether it is possible to pack these circles into a given container without overlapping one another or the container's boundary?

        \begin{figure}%
            \begin{tikzpicture}[scale=2.5]%
                \draw (0,0) rectangle (\B,\B);

\draw[->,thick] (-1,{\B/2}) -- node[above] {?} +(0.5,0);

\draw[filled] (-3,{\B/2+0.3}) circle(0.3257); 
\draw[filled] (-2.3,{\B/2+0.3}) circle(0.3257); 
\draw[filled] (-1.7,{\B/2+0.3}) circle(0.2303); 
\draw[filled] (-3,{\B/2-0.3}) circle(0.1629); 
\draw[filled] (-2.6,{\B/2-0.3}) circle(0.1152); 
\draw[filled] (-2.2,{\B/2-0.3}) circle(0.1152); 

%
            \end{tikzpicture}%
            \caption{Can these circles be packed into the square?}%
            \label{fig:big-question}%
        \end{figure}
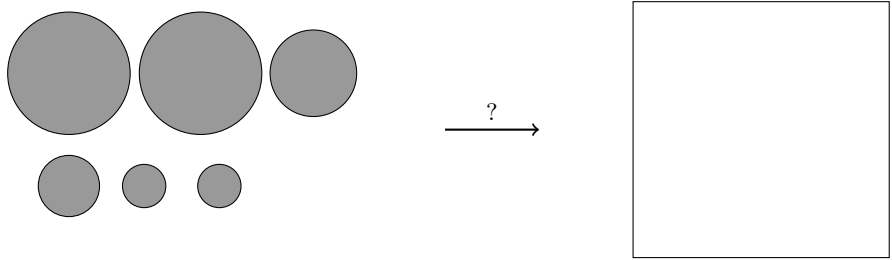%

This naturally occurring \emph{circle packing problem} has numerous applications in engineering, science, operational research and everyday life.
Examples include
packaging cylinders \cite{CKP2008solving,fraser1994integrated},
bundling tubes or cables \cite{WHZX2002improved,SSSKK2004disk},
the cutting industry \cite{SMCSCG2007new},
the layout of control panels \cite{CKP2008solving},
the design of digital modulation schemes \cite{PWMD1992packing},
or radio tower placement \cite{SMCSCG2007new}.
Further applications stem from chemistry \cite{WMP1994history},
foresting \cite{SMCSCG2007new},
and origami design \cite{lang1996computational}.

Despite their simple formulation, packing problems are quite difficult. In particular, deciding whether a given set of circles fits into a square container was shown to
be $\mathsf{NP}$-hard by Demaine, Fekete, and Lang in 2010~\cite{DFL2010circle}, using a reduction from \textsc{3-Partition}.
Their proof constructs a set of circles which first forces some symmetrical free “pockets” in the resulting circle packing.
The set's remaining circles can then be packed into these pockets if and only if the related \textsc{3-Partition} instance has a solution.
This means that there is (probably) no deterministic polynomial-time algorithm that can decide whether a given set of circles can be packed into a given container.
Additionally, due to the irrational coordinates which arise when packing circular objects, it is also surprisingly hard to solve circle packing problems in practice.
Even when the input consists of equally-sized circles, exact boundaries for the smallest square container are currently only known for up to 30 circles, and for 36 circles, see \cite{LR2002packing}.
For right isosceles triangular containers, optimal results have been published for up to 7 equal circles, see \cite{xu1996minimum}.

The related problem of packing square objects has also been studied for a long time.
The decision problem whether it is possible to pack a given set of squares into the unit square was shown to be strongly $\mathsf{NP}$-complete by Leung et al.~\cite{LTWYC1990packing}, also using a reduction from \textsc{3-Partition}.
Already in 1967, Moon and Moser~\cite{MM1967some} found a sufficient condition. They proved that it is possible to pack a set of squares into the unit square in a shelf-like manner if their combined area, the sum of all squares' areas, does not exceed $\nicefrac{1}{2}$, see \Cref{fig:shelf-packing}.

At the same time, $\nicefrac{1}{2}$ is the \emph{largest upper area bound} one can hope for, because two squares larger than the quarter-squares depicted in \Cref{fig:squares-worst} cannot be packed.
We call the ratio between the largest combined object area that can always be packed and the area of the container the problem's \emph{critical density}, or \emph{worst-case density}.

\begin{figure}
    \begin{minipage}{.45\textwidth}
        \centering
        \begin{tikzpicture}[scale=3.7]%
            \draw (0,0) rectangle (1,1);
\draw[filled] (0,0) rectangle (1/2,1/2);
\draw[filled] (1/2,1/2) rectangle (1,1);%
        \end{tikzpicture}%
    
        \caption{Worst-case instance for packing squares into a square.}
        \label{fig:squares-worst}
    \end{minipage}\hspace{10pt}%
    \begin{minipage}{.45\textwidth}
        \centering
        \begin{tikzpicture}[scale=3.7]%
            \input{tikz/shelf-packing.tex}%
        \end{tikzpicture}%
    
        \caption{Example packing produced by Moon and Moser's shelf-packing.}
        \label{fig:shelf-packing}
    \end{minipage}
\end{figure}

\begin{figure}
    \begin{minipage}{.45\textwidth}
        \centering
        \begin{tikzpicture}[scale=3.7]%
            \draw (0,0) rectangle (1,1);
\draw[filled] (0.2929,0.2929) circle (0.2929);
\draw[filled] (1-0.2929,1-0.2929) circle (0.2929);%
        \end{tikzpicture}%
    
        \caption{Worst case for packing circles into a square.}
        \label{fig:circles-worst}
    \end{minipage}\hspace{10pt}%
    \begin{minipage}{.45\textwidth}
        \centering
        \begin{tikzpicture}[scale=0.00925]%
            \draw (-200,-200) rectangle (200,200);
\draw[helper] (-157.12765267543034,-161.69665337894088) -- (161.69665337894088,157.12765267543034) arc (-45:90:25.11403961196304) (143.93834566633421,200) -- (-174.88596038803695,200) arc (90.00000000000001:180:25.11403961196304) -- (-200,174.88596038803695) -- (-200,-143.93834566633421) arc (-180:-45:25.11403961196304);
\draw[helper] (200,195.37820333164927) -- (-195.37820333164927,-200) arc (135:270:0) (-195.37820333164927,-200) -- (200,-200) arc (-90:0:0) -- (200,-200) -- (200,195.37820333164927) arc (0:135:0);
\draw[helper] (-200,91.47989903807937) -- (-200,-16.287562406340413) arc (-180:-45:77.98872528580094) (-66.86491820851447,-71.433918912025) -- (-12.981187486304627,-17.550188189815167) arc (-44.99999999999999:45:77.98872528580094) -- (-12.98118748630462,92.74252482155407) -- (-66.8649182085145,146.62625554376393) arc (45:180:77.98872528580094);
\draw[helper] (143.93834566633421,200) -- (-39.392327600528205,200) arc (90.00000000000001:225:25.11403961196304) (-57.15063531313485,157.12765267543034) -- (34.51470132029635,65.46231604199909) arc (-135.00000000000003:-45.000000000000014:25.11403961196304) -- (70.03131674550963,65.46231604199909) -- (161.69665337894085,157.1276526754303) arc (-45:90:25.11403961196304);
\draw[helper] (84.46660722207993,79.89760651856939) -- (161.69665337894088,157.12765267543034) arc (-45:90:25.11403961196304) (143.93834566633424,200) -- (66.7082995094733,200) arc (90.00000000000001:180:25.11403961196304) -- (41.594259897510256,174.88596038803695) -- (41.594259897510256,97.65591423117601) arc (-180:-45.00000000000001:25.11403961196304);
\draw[helper] (-57.15063531313484,157.1276526754303) -- (-2.085628531284158,102.06264589357963) arc (-135.00000000000003:-0:25.11403961196304) (40.78671879328553,119.82095360618624) -- (40.78671879328553,174.88596038803698) arc (0:89.99999999999999:25.11403961196304) -- (15.672679181322497,200) -- (-39.392327600528205,200) arc (90.00000000000001:225:25.11403961196304);
\draw[helper] (41.594259897510256,156.0931664175156) -- (41.594259897510256,113.74264197693027) arc (-180:-45:31.777380418458012) (95.84164149820454,91.272640794694) -- (117.01690371849719,112.44790301498665) arc (-45:45:31.777380418458012) -- (117.01690371849719,157.3879053794592) -- (95.84164149820455,178.56316759975186) arc (44.99999999999999:180:31.777380418458012);
\draw[helper] (143.93834566633424,200) -- (143.58694053928375,200) arc (90.00000000000001:225:25.11403961196304) (125.82863282667711,157.12765267543034) -- (126.00433539020236,156.9519501119051) arc (-135:-45:25.11403961196304) -- (161.52095081541563,156.9519501119051) -- (161.69665337894088,157.12765267543034) arc (-44.99999999999999:90:25.11403961196304);
\draw[helper] (200,-128.56011804735198) -- (200,102.8446106120408) arc (0:135:38.32866907957012) (134.56886910041072,129.94707243206) -- (18.86650477071432,14.244708102363596) arc (135:225:38.32866907957012) -- (18.86650477071432,-39.96021553767477) -- (134.56886910041067,-155.66257986737116) arc (-135.00000000000003:-0:38.32866907957012);
\draw[helper] (-195.37820333164927,-200) -- (177.71560807047825,-200) arc (-90:45:0) (177.71560807047825,-200) -- (-8.83129763058551,-13.453094298936236) arc (45:135:0) -- (-8.83129763058551,-13.453094298936236) -- (-195.37820333164927,-200) arc (135:270:0);
\draw[helper] (49.165988748867306,-70.25969951582776) -- (134.56886910041072,-155.6625798673712) arc (-135.00000000000003:-0:38.32866907957012) (200,-128.56011804735203) -- (200,-43.1572376958086) arc (0:90:38.32866907957012) -- (161.67133092042988,-4.828568616238471) -- (76.2684505688865,-4.828568616238471) arc (90.00000000000001:225:38.32866907957012);
\draw[helper] (134.5688691004107,129.94707243205997) -- (65.54662447898676,60.92482781063603) arc (135:270:38.32866907957012) (92.64908629900593,-4.506303088953274) -- (161.67133092042988,-4.506303088953274) arc (-90:0:38.32866907957012) -- (200,33.82236599061685) -- (200,102.8446106120408) arc (0:135:38.32866907957012);
\draw[helper] (134.3921094540351,-4.828568616238471) -- (76.26845056888648,-4.828568616238471) arc (90.00000000000001:225:38.32866907957012) (49.16598874886729,-70.25969951582775) -- (78.22781819144159,-99.32152895840206) arc (-135:-44.999999999999986:38.32866907957012) -- (132.43274183147997,-99.32152895840206) -- (161.49457127405427,-70.25969951582775) arc (-45:90:38.32866907957012);
\draw[helper] (200,-128.56011804735203) -- (200,-128.15645225720144) arc (0:135:38.32866907957012) (134.56886910041072,-101.05399043718225) -- (134.3670362053354,-101.25582333225753) arc (135:225:38.32866907957012) -- (134.3670362053354,-155.4607469722959) -- (134.56886910041072,-155.6625798673712) arc (-135:0:38.32866907957012);
\draw[helper] (-50.946262856713275,-55.568059525064) -- (-134.14135184680535,-138.76314851515608) arc (135:270:35.871717082794575) (-108.77621744475599,-200) -- (-25.581128454663933,-200) arc (-90:0:35.871717082794575) -- (10.29058862813065,-164.1282829172054) -- (10.29058862813065,-80.93319392711336) arc (0:135:35.871717082794575);
\draw[helper] (177.71560807047825,-200) -- (12.488963527027266,-34.77335545654901) arc (45:180:0) (12.488963527027266,-34.77335545654901) -- (12.488963527027266,-200) arc (-180:-90:0) -- (12.488963527027266,-200) -- (177.71560807047825,-200) arc (-90:44.99999999999999:0);
\draw[filled] (-0.4554878989983552,158.75779330771616) circle (41.242206692283894);
\draw[filled] (161.58772899795665,-128.3582851522767) circle (38.41227100204332);
\draw[filled] (143.762643102809,174.8131820032811) circle (25.186817996718908);
\draw[filled] (82.1427211144201,134.91790419722295) circle (40.54846121690985);
\draw[filled] (-99.69190265779991,37.596168315869456) circle (100.3080973422);
\draw[filled] (105.33028001146079,-55.19504159829659) circle (50.36647298205811);
\draw[filled] (-49.94840585844361,-139.76100551342574) circle (60.23899448657425);
\draw[filled] (141.4551835235315,54.03851338751521) circle (58.54481647646848);
\draw[filled] (60.88272728110479,-151.60623624592245) circle (48.39376375407751);%
        \end{tikzpicture}%
    
        \caption{Example packing produced by Split Packing.}
        \label{fig:split-packing}
    \end{minipage}
\end{figure}

The equivalent problem for packing circles has remained open.
When Demaine, Fekete, and Lang~\cite{DFL2010circle} posed the question in 2010, they suggested that the critical density may be determined by the two-circle instance shown in \Cref{fig:circles-worst}.
Again, it is easy to argue that if these two circles were only a little larger, we could no longer pack them into the unit square without overlap.
This means that their combined area constitutes an upper bound on the area that can always be packed.
In this paper, we show that indeed each set of circles of this total area can indeed be packed, but this requires a fundamentally different approach than Moon and Moser's orthogonal shelf-packing, see \Cref{fig:split-packing}.

We also study the problem of packing circles into non-acute triangular containers. It is obvious that circles larger than a triangle's incircle cannot be packed (compare \Cref{fig:triangle-worst}), but is it also possible to pack all sets of circles of up to that combined area?
We answer this question in the affirmative and introduce a weighted modification of the Split Packing algorithm, allowing us to pack circles into asymmetric non-acute triangles with critical density. See \Cref{fig:triangle-subcontainers} for an example packing.

\begin{figure}
    \begin{minipage}[t]{.45\textwidth}
        \centering
        \begin{tikzpicture}[scale=0.0065]%
            \draw (-400,-200) -- (400,-200) arc (-90:38.57867834990793:0) (400,-200) -- (90.5,188) arc (38.57867834990794:128.3450387527907:0) -- (90.5,188) -- (-400,-200) arc (128.34503875279077:270:0);
\draw[filled] (64.543277230317,-38.47870083291713) circle (161.52129916708287);%
        \end{tikzpicture}%
    
        \caption{Suspected worst-case instance for packing circles into a non-acute triangle.}
        \label{fig:triangle-worst}
    \end{minipage}\hspace{30pt}%
    \begin{minipage}[t]{.45\textwidth}
        \centering
        \begin{tikzpicture}[scale=0.0065]%
            \input{tikz/triangle-subcontainers.tex}%
        \end{tikzpicture}%
    
        \caption{Example packing produced by Split Packing.}
        \label{fig:triangle-subcontainers}
    \end{minipage}
\end{figure}

Many authors have considered heuristics for circle packing problems, see \cite{SMCSCG2007new,HM2009literature} for overviews of numerous heuristics and optimization methods.
The best known solutions for packing equal circles into squares, triangles and other shapes are continuously published on Specht's website \url{http://packomania.com} \cite{specht2015packomania}.

On the other hand, the literature on exact approximation algorithms which actually give performance guarantees is small.
Miyazawa et al.~\cite{MPSSW2014polynomial} devised asymptotic polynomial-time approximation schemes for packing circles into the smallest number of unit square bins.
More recently, Hokama, Miyazawa, and Schouery~\cite{HMS2016bounded} developed a bounded-space competitive algorithm for the online version of that problem.
As a byproduct of the tight worst-case bound, Split Packing yields an approximation algorithm for packing into single square and triangular containers. 

\subsection{Results}

We prove that the critical density for packing circles into a square is 

\begin{equation}
    \phi_s = \frac{\pi}{3+2\s} \approx 53.90\%\text{.}
\end{equation}

Any set of circles with a combined area of up to that percentage of the square's area can be packed, and for any higher percentage, there are sets that cannot be packed.

We also show that, for any right or obtuse triangle, any set of circles with a combined area of up to the triangle's incircle can be packed into that triangle.
At the same time, for any larger area, there are sets that cannot be packed, making the ratio between the incircle's and the triangle's area the packing problem's \emph{critical density}.
For a right isosceles triangle, this density is again approximately 53.90\%.
In the general case, the critical density for packing circles in a non-acute triangle with side lengths $x$, $y$, and $z$ is

\begin{equation}
    \phi_t = \pi\sqrt{\dfrac{(x+y-z)(z+x-y)(y+z-x)}{(x+y+z)^3}}.
\end{equation}

Our proofs are constructive: We describe a divide-and-conquer approach which repeatedly splits the set of circles in halves, and then packs these recursively, which is why we call this algorithm \emph{Split Packing}. In~\Cref{fig:example}, we demonstrate how the subdivision process looks like for an example set.

        \begin{figure}%
            \begin{tikzpicture}[scale=0.009]%
                \input{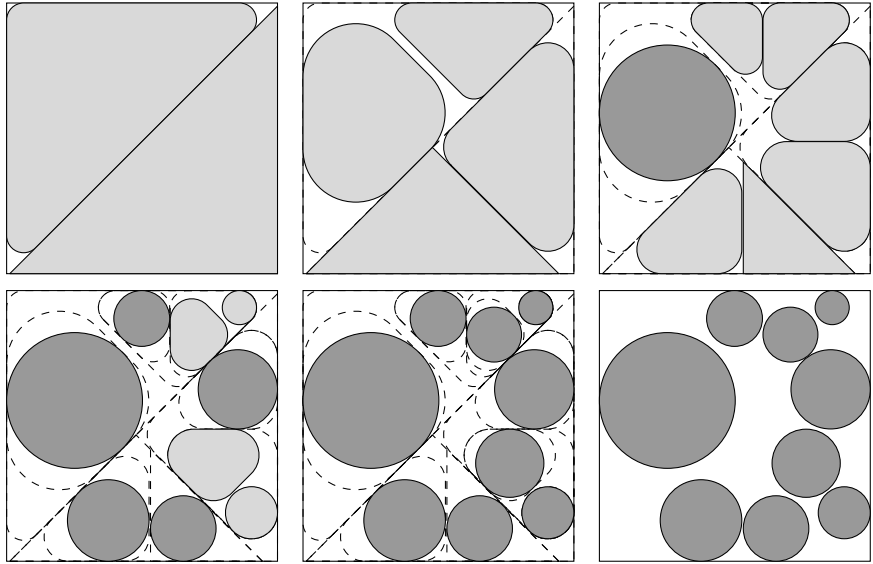}%
            \end{tikzpicture}%
            \caption{Split Packing recursively subdivides the container into subcontainers (light gray), before packing the circles into them (dark gray).}%
            \label{fig:example}%
        \end{figure}%

Split Packing can also be used as a constant-factor approximation algorithm for the smallest-area container of a given shape in which a given set of circles can be packed. For example, the ratio between the areas of the approximated and the optimal square is at most the reciprocal of the critical density, $\frac{3+2\s}{\pi} \approx 1.8552$.

While we focus on the problem of packing circles into square and triangular containers in this paper, we see more opportunities to generalize this approach in several directions, to allow other object and container shapes. We discuss some of these extensions in the conclusion.

\subsection{Key Ideas}

When we tried to prove that indeed all sets of circles with a combined area of up to the set shown in \Cref{fig:circles-worst} could be packed, many strategies proved unsuccessful.
But when we restricted the input set to only allow circles with areas equal to a negative power of two of the maximally packable area, we were surprised to find that these sets were easy to pack!

        \begin{figure}%
            \begin{tikzpicture}[scale=2.85]%
                    \draw (0,0) rectangle (\B,\B);
    \draw[filled] ({\B/2},{\B/2}) circle(0.5642) node {$a$};
\end{tikzpicture}
\begin{tikzpicture}[scale=2.85]
    \draw (0,0) rectangle (\B,\B);
    \draw (\B,0) -- (0,\B);
    \draw[filled] ({\B-0.3989},{\B-0.3989}) circle(0.3989) node {$\frac a 2$};
    \draw[filled] (0.3989,0.3989) circle(0.3989) node {$\frac a 2$};
\end{tikzpicture}
\begin{tikzpicture}[scale=2.85]
    \draw (0,0) rectangle (\B,\B);
    \draw (\B,0) -- (0,\B);
    \draw[filled] ({\B-0.3989},{\B-0.3989}) circle(0.3989) node {$\frac a 2$};
    \draw ({\B/2}, {\B/2}) -- (0,0);
    \draw[filled] (0.2820,{\B/2}) circle(0.2820) node {$\frac a 4$};
    \draw[filled] ({\B/2},0.2820) circle(0.2820) node {$\frac a 4$};
\end{tikzpicture}

\vspace{10pt}

\begin{tikzpicture}[scale=2.85]
    \draw (0,0) rectangle (\B,\B);
    \draw (\B,0) -- (0,\B);
    \draw[filled] ({\B-0.3989},{\B-0.3989}) circle(0.3989) node {$\frac a 2$};
    \draw ({\B/2}, {\B/2}) -- (0,0);
    \draw[filled] (0.2820,{\B/2}) circle(0.2820) node {$\frac a 4$};
    \draw ({\B/2}, {\B/2}) -- ({\B/2}, 0);
    \draw[filled] ({\B/2+0.1995},0.1995) circle(0.1995) node {$\frac a 8$};
    \draw[filled] ({\B/2-0.1995},0.1995) circle(0.1995) node {$\frac a 8$};
\end{tikzpicture}
\begin{tikzpicture}[scale=2.85]
    \draw (0,0) rectangle (\B,\B);
    \draw (\B,0) -- (0,\B);
    \draw[filled] ({\B-0.3989},{\B-0.3989}) circle(0.3989) node {$\frac a 2$};
    \draw ({\B/2}, {\B/2}) -- (0,0);
    \draw[filled] (0.2820,{\B/2}) circle(0.2820) node {$\frac a 4$};
    \draw ({\B/2}, {\B/2}) -- ({\B/2}, 0);
    \draw[filled] ({\B/2+0.1995},0.1995) circle(0.1995) node {$\frac a 8$};
    \draw ({\B/4}, {\B/4}) -- ({\B/2}, 0);
    \draw[filled] ({\B/2-0.1410},{\B/4}) circle(0.1410) node {$\frac a {16}$};
    \draw[filled] ({\B/4},0.1410) circle(0.1410) node {$\frac a {16}$};
\end{tikzpicture}
\begin{tikzpicture}[scale=2.85]
    \draw (0,0) rectangle (\B,\B);
    \draw (\B,0) -- (0,\B);
    \draw[filled] ({\B-0.3989},{\B-0.3989}) circle(0.3989) node {$\frac a 2$};
    \draw ({\B/2}, {\B/2}) -- (0,0);
    \draw[filled] (0.2820,{\B/2}) circle(0.2820) node {$\frac a 4$};
    \draw ({\B/2}, {\B/2}) -- ({\B/2}, 0);
    \draw ({\B/4}, {\B/4}) -- ({\B/2}, 0);
    \draw[filled] ({\B/2-0.1410},{\B/4}) circle(0.1410) node {$\frac a {16}$};
    \draw[filled] ({\B/4},0.1410) circle(0.1410) node {$\frac a {16}$};
    \draw ({\B/2+0.2}, 0) -- ({\B/4*3+0.1}, {\B/4-0.1});
    \draw[filled] ({\B/2+0.1995},0.1995) node {?};%
            \end{tikzpicture}%
            \caption{Splitting circles in half is easy.}%
            \label{fig:easy}%
        \end{figure}
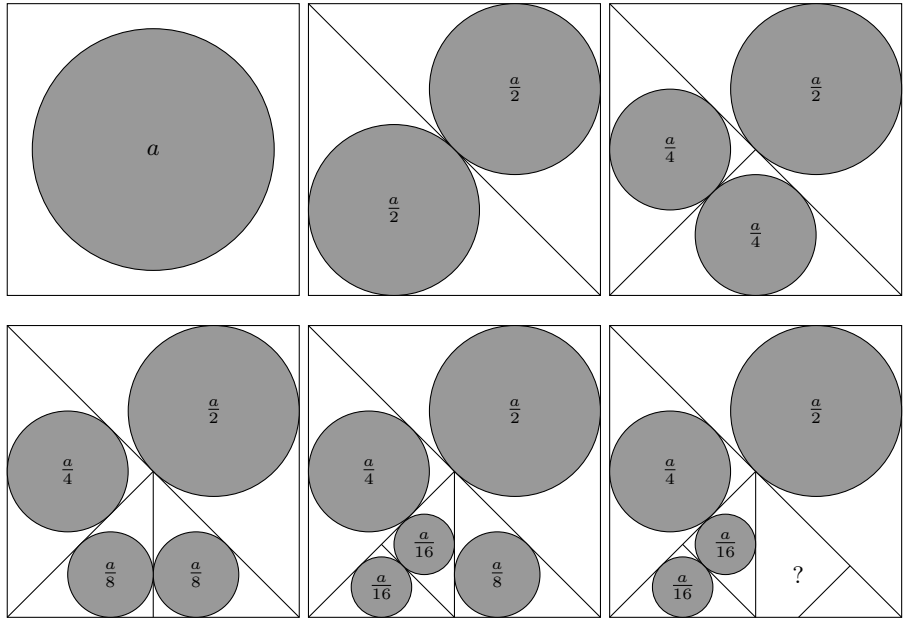%

Call the maximally packable area $a$.
A single circle of area $a$ can be packed, see the upper left of \Cref{fig:easy}. When splitting this circle into two equally-sized circles, and also cutting the square into halves along the circles' tangent, both circles now are incircles of isosceles right triangles. We can repeat this process of cutting one of the triangles into two smaller ones with half the area, and again the resulting circles are the triangles' incircles.
This allows us to recurse, and to repeat the splitting as often as necessary, until we arrive at the desired set of circles.
This divide-and-conquer approach to recursively split the set of input circles into subsets is the first key idea of the Split Packing algorithm.

For general sets of circles,
we could want to split a circle not exactly in half, but make one of the circles larger than the other one. In this case, we need to shift the cut to the side,
like in the lower right of \Cref{fig:easy}. For the small triangle, we can start another recursion, but it is unclear why we can continue with the recursion for the quadrilateral. Its shape resembles our triangles, but the lower corner is “cut off”.
We need an argument why these “degenerate triangles” do not break our packing strategy.

This is where the second key idea of the Split Packing algorithm comes into play.
When performing the top-level split, we could already decide which circles of the input set should go on which side of the cut-line.
With the power-of-two sets, it did not matter at all \emph{how} we split the circle set in half, as we could always split it into equally-sized halves.
But for general sets of circles, we need to proceed more carefully.
We perform the splitting of the set of circles into subgroups using an algorithm which resembles greedy scheduling.
This makes sure the resulting subgroups are \emph{close} to equal in terms of their combined area.
If the groups' areas deviate from the targeted 1:1 ratio, we can gain information about the minimum circle size in the larger group, allowing us to round off the subcontainer triangles.

Later in this paper, we also introduce a weighted generalization of the Split Packing approach:
When packing into asymmetric triangles, we do not want the resulting groups to have equal area, as it is not possible to cut the container into two subtriangles of equal size.
Instead, we target a different area ratio, defined by the incircles of the two triangles created by cutting the container orthogonally its the base through its tip, see \Cref{fig:hat-f}.
We call this desired area ratio the \emph{split key}.

The rest of the paper provides details for this process.

\section{Greedy Splitting}

The following definition makes it easier to refer to the properties of sets of circles.

\begin{definition}\label{def:circle-instance}
    A \emph{set of circles} is a multiset of positive real numbers, which define the circles' areas.
    For any set of circles~$C$, $\mysum(C)$ is the combined area of the set's circles and $\min(C)$ is the area of the smallest circle contained in the set.
\end{definition}

To differentiate between sets of circles and their elements, we will use upper case letters to denote sets of circles ($C$, $C_1$, $C_2$, \dots), while lower case letters will refer to their elements ($a$, $b$, $c$); when appropriate, we will also use these letters to refer to the respective areas.

The method which we use to split the sets of circles in half, \textsc{Split} (\Cref{alg:split}), resembles a greedy scheduling algorithm, which is why we call it \emph{greedy splitting}.
The circles are assumed to be sorted by size in descending order.
The algorithm first creates two empty “buckets”, and in each step adds the largest remaining circle of the input set to the more empty bucket.

If the resulting groups' areas deviate from the targeted 1:1 ratio, we gain additional information about the larger group: All its circles are at least as large as the group's area difference.

\begin{lemma}\label{th:min1}
    For any $C_1$ and $C_2$ produced by \textsc{Split}$(C)$,

    \begin{equation}
        \min(C_2) \ge \mysum(C_2) - \mysum(C_1)\text{.}
    \end{equation}
\end{lemma}

\begin{proof}
    Assume for contradiction the last element inserted into $C_2$ (let us call it~$c$) was smaller than $\mysum(C_2) - \mysum(C_1)$.
    This means that
    \begin{align*}
        \mysum(C_2) - c &> \mysum(C_2) - (\mysum(C_2) - \mysum(C_1))\\
        &= \mysum(C_1)\text{,}
    \end{align*}
    meaning that at the moment before $c$ was inserted, $\mysum(C_2)$ would already have been larger than~$\mysum(C_1)$.
    This is a contradiction, as the greedy algorithm would have put $c$ into the more empty group $C_1$ in this case.
    So $c$ must be at least $\mysum(C_2) - \mysum(C_1)$.
    Additionally, because the elements were inserted by descending size, all elements in $C_2$ must be at least as large as $c$.
\end{proof}

\begin{algorithm}
    \caption{\textsc{Split}$(C)$}
    \label{alg:split}
    \begin{algorithmic}
        \Require A set of circles $C$, sorted by size in descending order
        \Ensure Sets of circles $C_1$ and $C_2$
        \State $C_1 \gets \emptyset$
        \State $C_2 \gets \emptyset$
        \ForAll{$c \in C$}
            \If{$\mysum(C_1) \le \mysum(C_2)$}
                \State $C_1 \gets C_1 \cup \{c\}$
            \Else
                \State $C_2 \gets C_2 \cup \{c\}$
            \EndIf
        \EndFor
        \If{$\mysum(C_1) > \mysum(C_2)$}
            \State Swap $C_1$ and $C_2$
        \EndIf
    \end{algorithmic}
\end{algorithm}

\section{Split Packing}


In this section, we describe Split Packing as a rather general approach, before applying it to our concrete packing problem in the following sections.
This gives rise to the central Split Packing theorem.

To simplify talking about shapes which can pack certain classes of sets of circles, we introduce the following notions.

\begin{definition}
    $\C$ is the set of all sets of circles.
    $\C(a)$ is the set of exactly those sets of circles $C$ with $\mysum(C) \le a$.
    Finally, $\C(a,b)$ consists of exactly those sets of circles $C \in \C(a)$ with $\min(C) \ge b$.
\end{definition}

Let us give an example for the previous definition, as it is crucial for the rest of this paper. For any set of circles $C$ contained in $\C(1,\frac{1}{8})$, the combined area of $C$'s circles is at most~1, and at the same time, each of $C$'s circles has an area of at least $\frac{1}{8}$.

\begin{definition}
    For any $\mathcal{C} \subseteq \C$, a \emph{$\mathcal{C}$-shape} is a shape in which each $C \in \mathcal{C}$ can be packed.
\end{definition}

For example, if a shape is a $\C(a)$-shape, this means that it can pack all sets of circles with a combined area of $a$.
And a $\C(a,b)$-shape can pack all sets of circles with a combined area of $a$, whose circles each have an area of at least $b$.
With these preparations, we can now state our central theorem.

\begin{theorem}[Split Packing]\label{th:split-packing}
    A shape $s$ is a $\C(a,b)$-shape if for all $0 \le a_1 \le a_2$ with $a_1 + a_2 \le a$, one can find a $\C(a_1,b)$-shape and a $\C(a_2, \max\{a_2-a_1,b\})$-shape which can be packed into $s$.
\end{theorem}

\begin{proof}
    Consider an arbitrary $C \in \C(a,b)$.
    We use \textsc{Split}$(C)$ to produce two subsets $C_1$ and $C_2$.
    As $\min(C) \ge b$, all circles in the subsets will also have at least an area of $b$.
    Additionally, we know from \Cref{th:min1}, that $\min(C_2) \ge \mysum(C_2) - \mysum(C_1)$.
    So we can pack $C_1$ into the $\C(a_1,b)$-shape and $C_2$ into the $\C(a_2, \max\{a_2-a_1,b\})$-shape, and finally pack the two shapes into~$s$.
\end{proof}

Written as an algorithm, Split Packing looks as follows:

\begin{algorithm}
    \caption{\textsc{Splitpack}$(s,C)$}
    \begin{algorithmic}
        \Require A $\C(a,b)$-shape $s$ and a set of circles $C~\in~\C(a,b)$, sorted by size in descending order
        \Ensure A packing of $C$ into $s$
        \State $(C_1, C_2) \gets \textsc{Split}(C)$ \Comment{See \Cref{alg:split}}
        \State Determine a $\C(a_1,b)$-shape $s_1$
        \State \Call{Splitpack}{$s_1, C_1$}
        \State Determine a $\C(a_2,\max\{a_2-a_1,b\})$-shape $s_2$
        \State \Call{Splitpack}{$s_2, C_2$}
        \State Pack $s_1$, $s_2$, and their contents into $s$
    \end{algorithmic}
\end{algorithm}

This is a very general description of the Split Packing approach.
To apply it to concrete packing problems, one needs to show that all steps of the algorithm are always possible.

Note that the Split Packing algorithm can easily be extended to allow splitting into more than two subgroups.
For simplicity, we only describe the case of two subgroups here, as this suffices for the shapes we discuss in this paper.

\subsection{Analysis}\label{sec:analysis}

There are two perspectives on the implications of the Split Packing theorem.
Firstly, it gives a sufficient condition for the decision problem if a given set of circles can be packed into a given container: If the circles have a combined area of at most $a$, then the set can be packed.

Secondly, Split Packing can also be used as an approximation algorithm.
Suppose we are given a set of circles of combined area $a$ for which we want to find the smallest container of a certain shape (for example, triangular or square) in which the set can be packed.
We can then use Split Packing as an approximation algorithm, based on the critical density $d$ for the container.
We can then be sure that this container has at most $\frac{1}{d}$ times the area of the optimal container.

We first show that Split Packing has polynomial runtime, and then argue about the approximation factor.

\begin{lemma}
    Split Packing requires $\mathcal{O}(n)$ basic geometric constructions and $\mathcal{O}(n^2)$ numerical operations.
\end{lemma}

\begin{proof}
    Each subcontainer in the recursion tree either has two children (if more than one circle needs to be packed inside, in which case a \textsc{Split} is performed), or one (in this case, the child is a single circle and the recursion ends).
    Without the circles, the recursion tree is a full binary tree with $n$ leaf nodes, meaning that it has exactly $n-1$ interior nodes.
    The root node is the container of the packing problem, which does not need to be packed.
    In total, we need to pack $2n-2$ subcontainers, in addition to the $n$ circles of the input set, leading to $\mathcal{O}(n)$ geometric constructions.

    In addition, to build the recursion tree, we need at most a quadratic number of numeric operations:
    Before \textsc{Splitpack} is first invoked, the set of circles has to be sorted by size in descending order, this can be done in $\mathcal{O}(n \log n)$ time.
    Additionally, each run of the \textsc{Split} subroutine then takes linear time in the size of its input.
    If \textsc{Split} would partition its input into two subsets with a similar number of elements in each case, this would also lead to a runtime of $\mathcal{O}(n \log n)$.
    But in the worst case, each run only splits off one element, so that the total time needed for all \textsc{Split} operations is
    \begin{equation}
        t_{\text{\textsc{Split}}} = n + (n-1) + (n-2) + \dots + 1 \in \mathcal{O}(n^2).
    \end{equation}
\end{proof}

\begin{theorem}
    Split Packing, when used to pack circles into a $\C(a,b)$-shape of area $A$, is an approximation algorithm with an approximation factor of $\frac{A}{a}$, compared to the container of minimum area.
\end{theorem}

\begin{proof}
    We know from the previous lemma that Split Packing has polynomial runtime.
    As for the approximation factor, we can be sure that the area of the optimal container $\text{OPT}$ needs to be as least $a$, as we need to be able to fit the circles inside without overlap.
    At the same time, the area of the approximated container $\text{ALG}$ is exactly $A$, which means that
    \begin{equation}
        \frac{\text{ALG}}{\text{OPT}} \le \frac{A}{a}.
    \end{equation}
\end{proof}

\section{Packing into Right Hats}

After this general description of Split Packing, we now apply it to concrete containers.
We start with an observation.

If all circles which we want to pack have a certain minimum size, sharp corners of the container cannot be utilized anyway.
This observation motivates a family of shapes which resemble rounded triangles.
We call these shapes \emph{hats}.

\newcommand\defaulta{45}
\newcommand\defaultb{45}
\newcommand\defaultr{0.2}
\newcommand\defaultx{0.7}

        \begin{figure}%
            \begin{tikzpicture}[scale=3]%
                \pgfmathsetmacro{\a}{(\defaulta)}
\pgfmathsetmacro{\b}{(\defaultb)}
\pgfmathsetmacro{\round}{(0.3)}

\hatshape{1}{\round}{\a}{\b}

\draw[dashed] (rightcenter) circle(\s) node {$b$};
\draw[dashed] (leftcenter) circle(\s) node {$b$};
\draw[dashed] (topcenter) circle(\s) node {$b$};
\draw[dashed] (midcenter) circle(\r) node {$a$};%
            \end{tikzpicture}%
            \caption{A right $(a,b)$-hat.}%
            \label{fig:right-hat}%
        \end{figure}
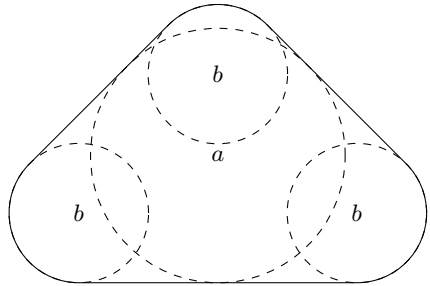%

\begin{definition}
    For each $0 \le b \le a$, a right \emph{$(a,b)$-hat} is an isosceles right triangle with an incircle of area $a$, whose three corners are rounded to the radius of a circle of area $b$, see \Cref{fig:right-hat}.
\end{definition}

We show that all sets of circles with a combined area of up to $a$ with a minimum circle size of $b$ can be packed into a right $(a,b)$-hat.
For the following proofs, we need to know a hat's dimensions in detail.
We construct these measures using \Cref{fig:hat-construction}.

\begin{lemma}\label{lm:sizes}
    Let $r$ be the radius of a circle of area $a$, and $s$ be the radius of a circle of area $b$.
    A right ($a,b$)-hat has
    \begin{itemize}
        \item $\begin{aligned}[t] \text{non-rounded height}\ h(a) = \textcolor{blue}{r + r\s} = \sqrt{\frac a \pi}(1+\s)\text{,} \end{aligned}$
        \item $\begin{aligned}[t]
                \text{width}\ w(a,b) &= 2\textcolor{blue}{(r+r\s)} - 2\textcolor{orange}{s\s}\\
                &= \sqrt{\frac a \pi}(2+2\s)-\sqrt{\frac b \pi}2\s\text{,}
            \end{aligned}$
        \item $\begin{aligned}[t]
                \text{diagonal}\ d(a,b) &= \textcolor{red}{(r+r\s)\s} - \textcolor{orange}{s\s}\\
                &= \sqrt{\frac a \pi}(2+\s)-\sqrt{\frac b \pi}\s\text{.}
            \end{aligned}$
    \end{itemize}

        \begin{figure}%
            \begin{tikzpicture}[scale=3]%
                \hatshape[draw=none]{1}{0.3}{45}{45}

\draw[dashed] (rightcenter) circle(\s);
\draw[dashed] (leftcenter) circle(\s);
\draw[dashed] (midcenter) circle(\r);
\draw[dashed] (right) -- (top) -- (left) -- cycle;

\draw[tight] (rightcenter) -- node[below left] {$s$} ++(-45:\s) -- node[below right] {$s$} ++(45:\s) coordinate (rightright) {} -- node[right] {$s$} ++(-90:\s) -- node[above] {$s$} (right);
\draw[thick,orange] (rightright) -- node[sloped,above] {$s\sqrt{2}$} (right);

\draw[tight] (leftcenter) -- node[above] {$s$} ++(180:\s) -- node[left] {$s$} ++(-90:\s) coordinate (leftleft) {} -- node[below left] {$s$} ++(135:\s) -- node[below right] {$s$} (left);
\draw[thick,orange] (leftleft) -- node[below] {$s\sqrt{2}$} (left);

\draw[thick,left,blue] (0,-\r) -- node {$r$} (midcenter) -- node {$r\sqrt{2}$} (top);
\draw (midcenter) -- node[below right] {$r$} +(45:\r) -- node[above right] {$r$} (top);
\draw[thick,red] (left) -- node[sloped,above] {$(r+r\sqrt{2})\sqrt{2}$} (top);

\coordinate (leftofleft) at ($(left)+(180:7pt)$);
\coordinate (belowleft) at ($(left)-(90:7pt)$);
\coordinate (farbelowleft) at ($(left)-(90:14pt)$);
\coordinate (abovetop) at ($(top)+(45:7pt)$);
\coordinate (farabovetop) at ($(top)+(45:14pt)$);
\coordinate (faraboveright) at ($(right)+(45:14pt)$);
\coordinate (aboverightright) at ($(rightright)+(45:7pt)$);
\coordinate (leftend) at ($(leftleft)-(90:7pt)$);
\coordinate (rightend) at ($(rightright)-(90:7pt)$);
\coordinate (totheright) at ($(rightcenter)+(0:\s)$);
\coordinate (rightend) at ($(totheright)-(90:7pt)$);

\draw[bracket] (leftofleft) -- node[left] {$h(a)$} (leftofleft |- top);
\draw[bracket] (rightend |- belowleft) -- node[below] {$w(a,b)$} (leftend |- belowleft);
\draw[bracket] (rightend |- farbelowleft) -- node[below] {$w'(a,b)$} (left |- farbelowleft);
\draw[bracket] (abovetop) -- node[above,sloped] {$d(a,b)$} (aboverightright);
\draw[bracket] (farabovetop) -- node[above,sloped] {$d'(a)$} (faraboveright);%
            \end{tikzpicture}%
            \caption{Constructing the dimensions of a right $(a,b)$-hat.}%
            \label{fig:hat-construction}%
        \end{figure}
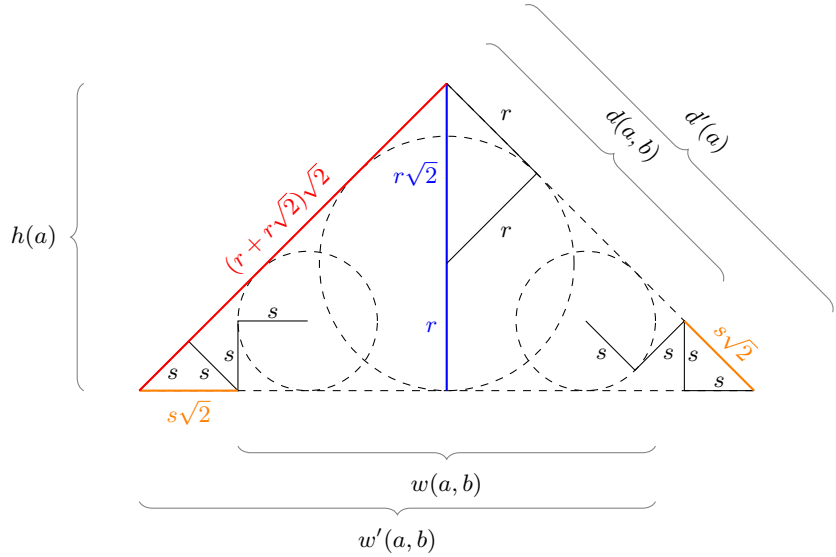%

    \noindent We define two additional measures for the case when one of the bottom corners is not rounded:
    \begin{itemize}
        \item $\text{corner-width}\ w'(a,b) = w(a,b) + \textcolor{orange}{s\s} = \sqrt{\dfrac a \pi}(2+2\s)-\sqrt{\dfrac b \pi}\s\text{,}$
        \item $\text{corner-diagonal}\ d'(a) = d(a,0) = \sqrt{\dfrac a \pi}(2+\s)\text{.}$
    \end{itemize}
\end{lemma}

\begin{lemma}\label{th:right-hats-in-right-hat}
    For each $0 \le a_1 \le a_2$, a right $(a_1,0)$-hat and a right $(a_2,a_2-a_1)$-hat can be packed into a right $(a_1+a_2,0)$-hat.
\end{lemma}

\begin{proof}
    Place the hats' tips at the bottom of the container hat and shift them to the left/right until their sides meet the sides of the container.
    \Cref{fig:right-hats-in-right-hat} illustrates how these packings looks for different ratios of $a_1$ and~$a_2$.
    This way of placing the two hats results in a valid packing if (1) the hats do not overlap each other and (2) the hats fit into the container hat individually.
    We will prove these two properties separately.

    \begin{figure}
        \begin{tikzpicture}[scale=2]
            \hatsinhat{\defaulta}{\defaultb}{0.5}{0}
        \end{tikzpicture}
        ~
        \begin{tikzpicture}[scale=2]
            \hatsinhat{\defaulta}{\defaultb}{0.55}{0}
        \end{tikzpicture}

        \vspace{10pt}

        \begin{tikzpicture}[scale=2]
            \hatsinhat{\defaulta}{\defaultb}{0.6}{0}
        \end{tikzpicture}
        ~
        \begin{tikzpicture}[scale=2]
            \hatsinhat{\defaulta}{\defaultb}{0.7}{0}
        \end{tikzpicture}

        \vspace{10pt}

        \begin{tikzpicture}[scale=2]
            \hatsinhat{\defaulta}{\defaultb}{0.9}{0}
        \end{tikzpicture}
        ~
        \begin{tikzpicture}[scale=2]
            \hatsinhat{\defaulta}{\defaultb}{1}{0}
        \end{tikzpicture}

        \caption{Hat-in-hat packings for different ratios of $a_1$ (incircle of left hat) and $a_2$ (incircle of right hat). The hats don't overlap horizontally if the sum of their corner-diagonals never gets larger than the container hat's width.}
        \label{fig:right-hats-in-right-hat}
    \end{figure}
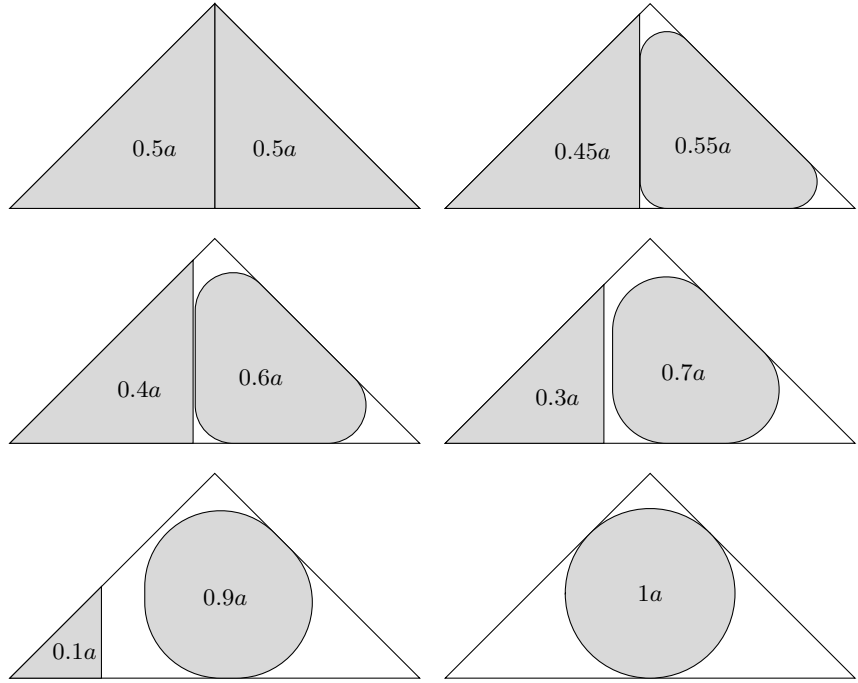

    \begin{itemize}
        \item[(1)]
            The hats do not overlap if the sum of their corner-diagonals is less or equal than the width of the $(a_1+a_2,0)$-hat.
            This can be verified to be true, as follows:
            \begin{align*}
                d'(a_1) + d'(a_2)
                &= \sqrt{\frac{a_1}{\pi}}(2+\s) + \sqrt{\frac{a_2}{\pi}}(2+\s)\\
                &= \frac{2+\s}{\sqrt\pi}(\sqrt{a_1} + \sqrt{a_2})\\
                &\le \frac{2+\s}{\sqrt\pi}(\sqrt{2a_1+2a_2})\\
                &= \sqrt{\frac{a_1+a_2}{\pi}}(2+2\s) = w(a_1+a_2,0).
            \end{align*}

        \item[(2)]
            The hats fit into the container hat individually if their corner-width never gets larger than the container hat's diagonal.
            For the $(a_1,0)$-hat, this is easy to show, as follows, using the fact that $w'(a,0)$ is a strictly increasing function on $a$:
            \begin{align*}
                w'(a_1,0) &\le w'(\frac{a_1+a_2}{2},0) = \sqrt{\frac{a_1+a_2}{2\pi}}(2+2\s)\\
                &= \sqrt{\frac{a_1+a_2}{\pi}}(2 + \s) = d(a_1+a_2,0).
            \end{align*}

            For the $(a_2,a_2-a_1)$-hat, we need to show that the following inequality holds:
            \begin{align*}
                w'(a_2,a_2-a_1) \le d(a_1+a_2,0).
            \end{align*}

            Let $a = a_1+a_2$.
            It then suffices to show that for all $0 \le a_1 \le a/2$,
            \begin{align*}
                w'(a-a_1,a&-2a_1)\\
                &= \sqrt{\frac{a-a_1}{\pi}}(2+2\s) - \sqrt{\frac{a-2a_1}{\pi}}\sqrt{2}\\
                &\le \sqrt{\frac{a}{\pi}}(2+2\s) = d(a,0).
            \end{align*}

            The left expression has its only extremum at $a_1 = \frac{1}{4}(3-\s)a \approx 0.3964a$.
            This point turns out to be a global minimum.
            As we can check the inequality to be true for $a_1 = 0$ and $a_1 = a/2$, it always holds between those two values.
    \end{itemize}
\end{proof}

\noindent The next lemma extends this observation to rounded container hats.

\begin{lemma}\label{th:rounded-right-hats-in-right-hat}
    For each $0 \le a_1 \le a_2$, a right $(a_1,b)$-hat and a right $(a_2,\max\{a_2-a_1,b\})$-hat can be packed into a right $(a_1+a_2,b)$-hat.
\end{lemma}

\begin{proof}
    \Cref{th:right-hats-in-right-hat} tells us that \Cref{th:rounded-right-hats-in-right-hat} is true for $b = 0$.
    Now the container's corners can be rounded to the radius of a circle of area $b$, and we need to show that the two hats from the previous construction still fit inside.
    But all of the two hat's corners are also rounded to (at least) the same radius, so they will never overlap the container, see \Cref{fig:right-hats-rounding}.
\end{proof}

        \begin{figure}%
            \begin{tikzpicture}[scale=2]%
                    \hatsinhat{\defaulta}{\defaultb}{\defaultx}{0}
\end{tikzpicture}
\begin{tikzpicture}[scale=2]
\draw[->,thick] (0,0) -- (0.6,0);
\draw[draw=none] (0,0) -- (0,-0.6);
\end{tikzpicture}
\begin{tikzpicture}[scale=2]
    \hatsinhat{\defaulta}{\defaultb}{\defaultx}{\defaultr}%
            \end{tikzpicture}%
            \caption{Rounding all hats' corners by the same radius does not affect the packing.}%
            \label{fig:right-hats-rounding}%
        \end{figure}
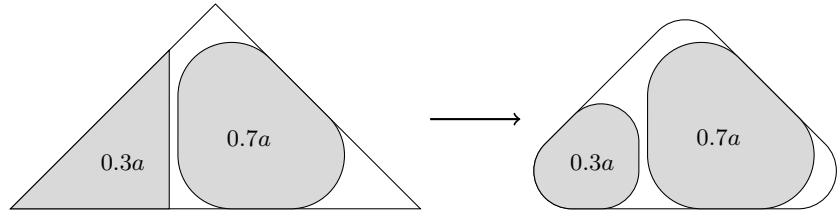%

With these preparations, we can finally apply Split Packing to right hats.

\begin{theorem}\label{th:right-hats}
    Given a right $(a,b)$-hat, all sets of circles with a combined area of at most $a$ and a minimum circle size of at least $b$ can be packed into that hat.
\end{theorem}

\begin{proof}
    We prove by induction that we can pack each $C \in \C(a,b)$ into the hat.
    If $C$ only consists of a single circle, it can be packed into the hat, as it is at most as big as the hat's incircle.

    Now assume that for any $0 \le b \le a$, any right $(a,b)$-hat could pack all sets of circles into $\C(a,b)$ with at most $n$ circles.
    Consider a set of circles $C \in \C(a,b)$ containing $n+1$ circles.

    \textsc{Split} will partition $C$ into two subsets $C_1 \in \C(a_1, b)$ and $C_2 \in \C(a_2, \max\{(b_2-b_1,b\})$
    As \textsc{Split} can never return an empty set (except for $|C| = 1$, a case which we handled above), each subset will contain at most $n$ circles.
    We know from \Cref{th:rounded-right-hats-in-right-hat} that we can find two hats with matching parameters which fit into the container hat.
    By assumption, these hats can now pack all sets from $\C(a_1, b)$ and $\C(a_2, \max\{(b_2-b_1,b\})$, respectively, which means that they can especially also pack $C_1$ and $C_2$.
    If we then pack the two hats into the container, we have constructed a packing of $C$ into the container hat.
    By induction, we can pack each $C \in \C(a,b)$ into the $(a,b)$-hat.
\end{proof}

\section{Packing into Squares}

With these preparations, we turn to square containers.
Having established right $(a,b)$-hats as $\C(a,b)$-shapes, to argue about the packing properties of squares is going to be relatively straightforward.
We first argue about the worst-case instance for squares.

To simplify talking about a square's worst-case instance, we introduce the following notion in analogy to the \emph{incircle}:

\begin{definition}
    A shape's \emph{twincircles} are the largest two equal circles that can be packed into the shape.
\end{definition}

\begin{lemma}\label{th:square-worst}
    Two touching equal circles, packed into opposing corners of a square, are the square's \emph{twincircles}, meaning that there are no two larger equal circles which can be packed.
\end{lemma}

\begin{proof}
    Let $r$ be the radius of these circles.
    When eroding the square by $r$, the result is a square with a diagonal of $2r$.
    When eroding by a larger radius $r + \varepsilon$, the diagonal will be smaller than $2r$.
    But the centers of the two circles need to be placed at least $2r + 2\varepsilon$ away from each other, and additionally need a distance of at least $r + \varepsilon$ from the square's boundary.
    Both constraints cannot be satisfied at the same time.
\end{proof}

\begin{lemma}\label{th:square-twincircle-area}
    The twincircles of a square with area $a$
    have a combined area of
    \begin{equation}
        a_{tc} = \frac{\pi}{3+2\s}a \approx 0.5390a.
    \end{equation}
\end{lemma}

\begin{proof}
    We can construct the twincircles' radius $r$ as seen in \Cref{fig:square-worst-construction}:
    \begin{equation*}
        2r + 2\frac{r}{\s} = \sqrt{a} \iff r = \frac{\sqrt{a}}{2 + \s}.
    \end{equation*}

    So the combined area of the twincircles is
    \begin{equation*}
        a_{tc} = 2\pi r^2 = 2\pi\frac{a}{4 + 4\s + 2} = \frac{\pi}{3+2\s}a. \eqno
    \end{equation*}
\end{proof}

        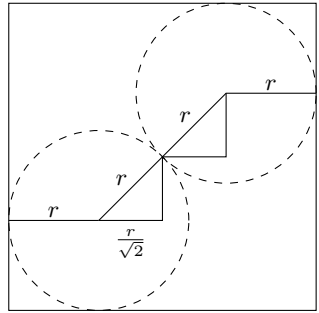
\begin{figure}%
            \begin{tikzpicture}[scale=3]%
                \draw (0,0) rectangle (\B,\B);
\newcommand\circrad{0.398942}
\coordinate (c1) at (\circrad,\circrad);
\coordinate (c2) at (\B-\circrad,\B-\circrad);
\coordinate (c) at (\B/2,\B/2);
\draw[dashed] (c1) circle (\circrad);
\draw[dashed] (c2) circle (\circrad);
\draw[tight] (c1) edge[above] node {$r$} +(180:\circrad) -- node[above left] {$r$} ++(45:\circrad) -- node[above left] {$r$} ++(45:\circrad) -- node[above] {$r$} ++(0:\circrad);
\draw (c1) -- node[below] {$\frac{r}{\s}$} (c1 -| c) -- (c) -- (c2 |- c) -- (c2);%
            \end{tikzpicture}%
            \caption{Constructing the twincircles' radius $r$.}%
            \label{fig:square-worst-construction}%
        \end{figure}%

We now proceed in analogy to \Cref{th:right-hats-in-right-hat}:

\begin{lemma}\label{th:hats-in-square}
    For each $0 \le a_1 \le a_2$, a right $(a_1,0)$-hat and a right $(a_2,a_2-a_1)$-hat can be packed into a square with a twincircle area of $a_1+a_2$.
\end{lemma}

\begin{figure}
    \begin{tikzpicture}[scale=2]
        \hatsinsquare{0.5}
    \end{tikzpicture}
    ~
    \begin{tikzpicture}[scale=2]
        \hatsinsquare{0.48}
    \end{tikzpicture}
    ~
    \begin{tikzpicture}[scale=2]
        \hatsinsquare{0.45}
    \end{tikzpicture}
    ~
    \begin{tikzpicture}[scale=2]
        \hatsinsquare{0.4}
    \end{tikzpicture}

    \vspace{10pt}

    \begin{tikzpicture}[scale=2]
        \hatsinsquare{0.3}
    \end{tikzpicture}
    ~
    \begin{tikzpicture}[scale=2]
        \hatsinsquare{0.2}
    \end{tikzpicture}
    ~
    \begin{tikzpicture}[scale=2]
        \hatsinsquare{0.1}
    \end{tikzpicture}
    ~
    \begin{tikzpicture}[scale=2]
        \hatsinsquare{0}
    \end{tikzpicture}

    \caption{Hat-in-square packings for different ratios of $a_1$ and $a_2$}
    \label{fig:hats-in-square}
\end{figure}
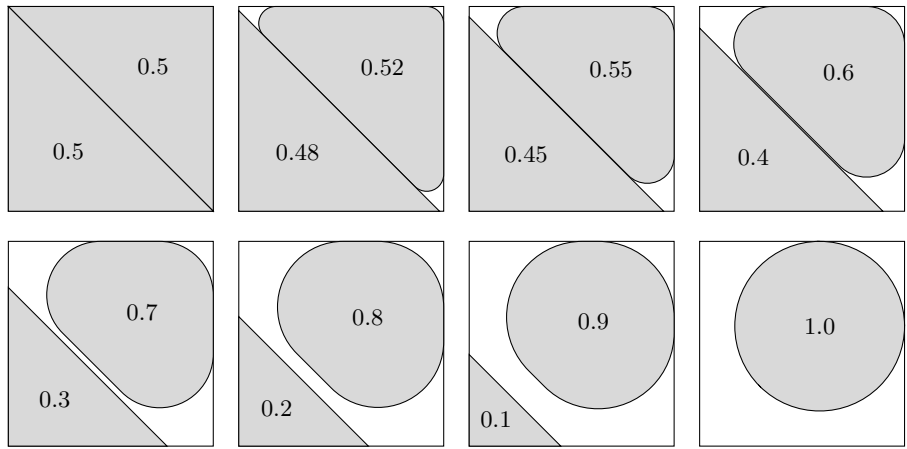

\begin{proof}
    Place the hats' tips in two opposing corners of the square, like in \Cref{fig:hats-in-square}.
    Again, this placement constitutes a valid packing because (1) the hats never overlap and (2) the hats fit into the square individually.
    We can prove both properties in a similar fashion as in \Cref{th:right-hats-in-right-hat}:

    \begin{itemize}
        \item[(1)]
            The hats do not overlap if their combined height never exceeds the square's diagonal, $\frac{\s+2}{\sqrt\pi}\sqrt{a_1+a_2}$, which is the case:
            \begin{align*}
                h(a_1) + h(a_2)
                &= \sqrt{\frac{a_1}{\pi}}(1+\s) + \sqrt{\frac{a_2}{\pi}}(1+\s)\\
                &= \frac{1+\s}{\sqrt\pi}(\sqrt{a_1}+\sqrt{a_2})\\
                &\le \frac{1+\s}{\sqrt\pi}(\sqrt{2a_1+2a_2})\\
                &= \frac{\s+2}{\sqrt\pi}\sqrt{a_1+a_2}.
            \end{align*}

        \item[(2)]
            Again, let $a = a_1 + a_2$.
            The hats fit into the square individually if their diagonal never gets larger than the square's edge length $\frac{1+\s}{\sqrt{\pi}}\sqrt{a}$.
            For the smaller hat, this is easy to show:
            \begin{align*}
                d(a_1,0) \le d(\frac{a_1+a_2}{2},0)
                &= \sqrt{\frac{\frac{a_1+a_2}{2}}{\pi}}(2+\s)\\
                &= \frac{1+\s}{\sqrt{\pi}}\sqrt{a}.
            \end{align*}

            For the larger hat, we need to show that the following inequality holds:
            \begin{equation*}
                d(a_2,a_2-a_1) \le \frac{1+\s}{\sqrt{\pi}}\sqrt{a}.
            \end{equation*}

            As $a_2$ is smaller than $a-a_1$, it suffices to show that for all $0 \le a_1 \le a/2$,
            \begin{align*}
                d(a-a_1,a&-2a_1)\\
                &= \sqrt{\frac{a-a_1}{\pi}}(2+\s)-\sqrt{\frac{a-2a_1}{\pi}}\s\\
                &\le \frac{1+\s}{\sqrt{\pi}}\sqrt{a}.
            \end{align*}

            The left expression has its only extremum at $a_1 = \frac{1}{14}(9-4\s)a \approx 0.2388a$.
            This point turns out to be a global minimum.
            As we can check the inequality to be true for $a_1 = 0$ and $a_1 = a/2$, it always holds between those two values.
    \end{itemize}
\end{proof}

We are now ready to prove our main result:

\begin{theorem}\label{th:square}
    Given a square with a twincircle area of $a$, all sets of circles with a combined area of up to $a$ can be packed into the square, and this area bound is tight.
    Expressed algebraically, the critical density is
    \begin{equation}
        \phi_s = \frac{\pi}{3+2\s} \approx 0.5390.
    \end{equation}
\end{theorem}

\begin{proof}
    By \Cref{th:hats-in-square} and the Split Packing Theorem (\Cref{th:split-packing}), the square is a $\C(a)$-shape.

    On the other hand, as shown in \Cref{th:square-worst}, two equal circles with a combined area of more than $a$ cannot be packed.
    We know from \Cref{th:square-twincircle-area} that the ratio between the twincircles' area and the square's area is $\frac{\pi}{3+2\s}$.
\end{proof}

\begin{figure}
        \begin{tikzpicture}[scale=0.009]%
            \draw (-200,-200) rectangle (200,200);
\draw[filled] (-34.31457505076196,34.31457505076196) circle (165.68542494923804);%
        \end{tikzpicture}%
    
    ~
        \begin{tikzpicture}[scale=0.009]%
            \draw (-200,-200) rectangle (200,200);
\draw[filled] (-82.84271247461902,82.84271247461902) circle (117.15728752538101);
\draw[filled] (82.84271247461902,-82.84271247461902) circle (117.15728752538101);%
        \end{tikzpicture}%
    
    ~
        \begin{tikzpicture}[scale=0.009]%
            \draw (-200,-200) rectangle (200,200);
\draw[filled] (-125.90322538651282,88.9009414228247) circle (74.09677461348718);
\draw[filled] (125.9032253865128,-21.114561800016883) circle (74.09677461348718);
\draw[filled] (59.29260780414956,125.90322538651277) circle (74.09677461348718);
\draw[filled] (21.114561800016833,-125.90322538651283) circle (74.09677461348718);
\draw[filled] (-125.9032253865128,-59.2926078041496) circle (74.09677461348718);%
        \end{tikzpicture}%

    \vspace{10pt}

        \begin{tikzpicture}[scale=0.009]%
            \draw (-200,-200) rectangle (200,200);
\draw[filled] (-82.84271247461896,82.84271247461896) circle (117.15728752538105);
\draw[filled] (117.15728752538094,-0) circle (82.84271247461905);
\draw[filled] (-58.57864376269019,-141.42135623730945) circle (58.578643762690525);
\draw[filled] (99.99999999999969,-158.5786437626905) circle (41.421356237309524);
\draw[filled] (29.289321881345618,-70.7106781186554) circle (29.289321881345263);
\draw[filled] (20.710678118655107,-149.99999999999915) circle (20.710678118654762);
\draw[filled] (64.64466094067159,-114.64466094067322) circle (14.644660940672631);
\draw[filled] (25.00000000000183,-110.35533905932797) circle (10.355339059327381);
\draw[filled] (42.6776695296627,-132.3223304703341) circle (7.322330470336316);
\draw[filled] (44.822330470335324,-112.50000000000347) circle (5.1776695296636905);
\draw[filled] (33.838834764836534,-121.33883476482961) circle (3.661165235168158);
\draw[filled] (43.749999999993335,-122.41116523516594) circle (2.5888347648318453);
\draw[filled] (39.33058261758878,-116.91941738242508) circle (1.830582617584079);
\draw[filled] (38.79441738242062,-121.87499999998639) circle (1.2944173824159226);
\draw[filled] (41.54029130877401,-119.66529130880119) circle (0.9152913087920395);
\draw[filled] (39.06250000002744,-119.39720869121709) circle (0.6472086912079613);
\draw[filled] (40.167354345585956,-120.77014565435967) circle (0.45764565439601973);
\draw[filled] (40.301395654378,-119.53125000005464) circle (0.32360434560398066);
\draw[filled] (39.614927172874914,-120.08367717276566) circle (0.22882282719800986);
\draw[filled] (40.234374999891024,-120.15069782716168) circle (0.16180217280199033);
\draw[filled] (39.958161413671924,-119.80746358654659) circle (0.11441141359900493);
\draw[filled] (39.92465108647392,-120.11718749978179) circle (0.08090108640099516);
\draw[filled] (40.096268206508626,-119.9790807069451) circle (0.057205706799502466);
\draw[filled] (39.94140625043672,-119.96232554334608) circle (0.04045054320049758);
\draw[filled] (40.01045964630937,-120.04813410281773) circle (0.028602853399751233);
\draw[filled] (40.01883722810888,-119.97070312587317) circle (0.02022527160024879);
\draw[filled] (39.97593294946444,-120.0052298227181) circle (0.014301426699875617);
\draw[filled] (40.014648435753934,-120.00941861361784) circle (0.010112635800124396);
\draw[filled] (39.99738508951426,-119.98796647647843) circle (0.007150713349937808);
\draw[filled] (39.99529069406437,-120.00732421525758) circle (0.005056317900062198);
\draw[filled] (40.00601675826852,-119.99869254650335) circle (0.003575356674968904);
\draw[filled] (39.99633789761008,-119.9976453487784) circle (0.002528158950031099);
\draw[filled] (40.00065372325606,-120.00300837214932) circle (0.001787678337484452);
\draw[filled] (40.00117732211853,-119.99816895928242) circle (0.0012640794750155494);
\draw[filled] (39.99849582789541,-120.00032685464308) circle (0.000893839168742226);
\draw[filled] (40.000915499404094,-120.0005886540743) circle (0.0006320397375077747);
\draw[filled] (39.9998365866487,-119.99924794188769) circle (0.000446919584371113);
\draw[filled] (39.99970568693308,-120.00045770779151) circle (0.00031601986875388736);
\draw[filled] (40.000375973174535,-119.99991832126567) circle (0.0002234597921855565);
\draw[filled] (39.99977122993167,-119.99985287140784) circle (0.00015800993437694368);
\draw[filled] (40.00004078347489,-120.00018787480889) circle (0.00011172989609277825);
\draw[filled] (40.00007350840379,-119.99988578266952) circle (0.00007900496718847184);
\draw[filled] (39.99990628627076,-120.00002027987365) circle (0.00005586494804638913);
\draw[filled] (40.00005677286338,-120.00003664233812) circle (0.00003950248359423592);
\draw[filled] (39.99999008442472,-119.99995359143499) circle (0.000027932474023194564);
\draw[filled] (39.99998190319249,-120.00002771164203) circle (0.00001975124179711796);
\draw[filled] (40.00002229996034,-119.99999549610641) circle (0.000013966237011597282);
\draw[filled] (39.99998752017844,-119.99999140549029) circle (0.00000987562089855898);
\draw[filled] (40.0000012992339,-120.00000927516186) circle (0.000006983118505798641);
\draw[filled] (40.00000334454196,-119.99999675952098) circle (0.00000493781044927949);%
        \end{tikzpicture}%
    
    ~
        \begin{tikzpicture}[scale=0.009]%
            \draw (-200,-200) rectangle (200,200);
\draw[filled] (88.40531002041091,-121.54823284026168) circle (0);
\draw[filled] (88.40531002323715,-121.54823284708458) circle (0);
\draw[filled] (88.40530988620473,-121.54823279045964) circle (0);
\draw[filled] (88.40531099627273,-121.54823232378726) circle (0);
\draw[filled] (88.4053087682328,-121.54823793675526) circle (0.00000274245604817574);
\draw[filled] (88.4052995308953,-121.54821735435358) circle (0.000010764627881780642);
\draw[filled] (88.40536463199265,-121.54824079554076) circle (0.00003420581509076833);
\draw[filled] (88.40519950973234,-121.54832260159922) circle (0.00009311926642302679);
\draw[filled] (88.40533165014672,-121.54790157203205) circle (0.00022525968080586715);
\draw[filled] (88.40580010121054,-121.54875352654855) circle (0.0004964375728054084);
\draw[filled] (88.40386585409499,-121.54823532925458) circle (0.0010146348667689656);
\draw[filled] (-30.568836741060213,115.8002307767114) circle (0.0019486057912721595);
\draw[filled] (88.41061721296475,-121.54507196290719) circle (0.0035517287101748812);
\draw[filled] (-30.564593432916976,115.81192055055763) circle (0.00619191393451403);
\draw[filled] (88.40378053155817,-121.56472089494943) circle (0.010388410116754329);
\draw[filled] (88.38606866213362,-121.5378619237992) circle (0.01685638777519079);
\draw[filled] (-30.54155614501455,115.76007813784132) circle (0.02656021302334192);
\draw[filled] (88.46654312539131,-121.56178189386739) circle (0.04077635784339093);
\draw[filled] (-30.653267797599543,115.79468485019163) circle (0.061166925373672645);
\draw[filled] (88.31097314602486,-121.64793033288198) circle (0.08986479978662265);
\draw[filled] (88.35067980674614,-121.38528736594074) circle (0.12957146050786764);
\draw[filled] (-30.317102681419048,115.95760653591137) circle (0.18366852984848506);
\draw[filled] (88.64144763614948,-121.85533885450926) circle (0.2563441516594916);
\draw[filled] (87.86230603697051,-121.75894767973855) circle (0.35273532643018807);
\draw[filled] (-30.612521506963247,115.0592024499691) circle (0.47908735539268693);
\draw[filled] (-31.39898713865972,116.07737879591255) circle (0.6429315727102998);
\draw[filled] (-29.27935491352397,115.86702779786401) circle (0.8532825707588387);
\draw[filled] (90.03415045871537,-120.67339014266358) circle (1.1208561488942967);
\draw[filled] (-32.18109166867866,113.82092796532088) circle (1.4583092409710623);
\draw[filled] (89.27450350634787,-124.23460332682801) circle (1.8805031012618034);
\draw[filled] (-31.2346098572891,119.23424597190974) circle (2.40479105236062);
\draw[filled] (85.60144559589635,-118.90574313581673) circle (3.051332122149742);
\draw[filled] (94.99922324102728,-119.69784293366651) circle (3.8434319199995275);
\draw[filled] (-25.802217858877448,110.40332756044442) circle (4.807912125072467);
\draw[filled] (82.76992295068553,-128.9119295127834) circle (5.975509981933292);
\draw[filled] (161.37459830033742,192.61869077935336) circle (7.381309220646655);
\draw[filled] (178.11466194730798,190.9347961598127) circle (9.065203840187307);
\draw[filled] (87.86680918406243,-103.97142529217268) circle (11.072396215310189);
\draw[filled] (-52.25607221155,119.04934644341525) circle (13.453931008043323);
\draw[filled] (102.92752770050083,35.56729913655748) circle (16.267266385687623);
\draw[filled] (66.56281177995366,38.87691681851679) circle (19.57688406764693);
\draw[filled] (-13.847040337998333,143.3148161257596) circle (23.454939743579683);
\draw[filled] (-18.374056019690222,77.60972023630171) circle (27.98195542527156);
\draw[filled] (119.28741783205234,-166.75244468571455) circle (33.247555314285464);
\draw[filled] (-85.33448295839592,160.6487532131388) circle (39.35124678686122);
\draw[filled] (153.59675188327952,87.97264905448485) circle (46.403248116720455);
\draw[filled] (64.13352682140095,141.24403689329145) circle (54.52536457539092);
\draw[filled] (136.14808543159384,-46.29322860143612) circle (63.85191456840612);
\draw[filled] (-20.06695517425086,-125.46929251573094) circle (74.53070748426906);
\draw[filled] (-113.275925048602,9.370437927961033) circle (86.724074951398);%
        \end{tikzpicture}%
    
    ~
        \begin{tikzpicture}[scale=0.009]%
            \draw (-200,-200) rectangle (200,200);
\draw[filled] (106.3812562905312,62.83713513995149) circle (4.640120015402711);
\draw[filled] (-104.4607353927466,-80.91522198331653) circle (6.562120656821369);
\draw[filled] (109.77805989502328,77.0410071537171) circle (8.036923619894777);
\draw[filled] (-22.6399140334728,-9.67250471958291) circle (9.280240030805421);
\draw[filled] (-70.24210909597723,110.21457367292807) circle (10.375623778197834);
\draw[filled] (56.311580322034104,-123.08920607852684) circle (11.365926383011862);
\draw[filled] (-87.30292687116163,-134.58417881963717) circle (12.276603614248772);
\draw[filled] (73.21092791447367,146.37099813831458) circle (13.124241313642738);
\draw[filled] (-141.23067894633206,78.69062417162694) circle (13.920360046208133);
\draw[filled] (79.98904311330476,-59.551521071228805) circle (14.67334786520815);
\draw[filled] (184.6104629266911,-163.16009639329374) circle (15.389537073308913);
\draw[filled] (-81.84464153562351,59.739503064342685) circle (16.073847239789554);
\draw[filled] (16.495249492346264,160.85830580899352) circle (16.73019063984123);
\draw[filled] (37.48643553239152,-182.63826066885085) circle (17.36173933114917);
\draw[filled] (-54.35131767361856,-79.92319520785811) circle (17.971107544058402);
\draw[filled] (-117.7368961101848,144.6130761152724) circle (18.560480061610843);
\draw[filled] (-149.82270355300426,-122.9232189566351) circle (19.131704939048024);
\draw[filled] (151.1617640807145,123.16415047373154) circle (19.686361970464105);
\draw[filled] (179.77418576695922,151.0136714539874) circle (20.22581423304081);
\draw[filled] (-179.24875244360427,-150.05888871294087) circle (20.75124755639567);
\draw[filled] (52.35656182475727,178.7362987957374) circle (21.263701204262585);
\draw[filled] (21.6855528732591,-143.01000662047363) circle (21.764092047717014);
\draw[filled] (-22.331773016280597,-53.9594916366128) circle (22.253233841822624);
\draw[filled] (-144.49275537659506,177.2681472339763) circle (22.731852766023724);
\draw[filled] (-176.79939992298642,56.24614450861535) circle (23.20060007701355);
\draw[filled] (56.728670636385495,-23.582014871022757) circle (23.660062503950876);
\draw[filled] (-57.191527999998215,23.72046314570552) circle (24.110770859684333);
\draw[filled] (160.34117510666113,-125.93142195868252) circle (24.553207228497545);
\draw[filled] (-64.85127858876153,-175.0121889921729) circle (24.98781100782709);
\draw[filled] (42.29316603364688,122.77823668352694) circle (25.414984019672836);
\draw[filled] (-26.07201147463596,132.52167762183672) circle (25.83509486101351);
\draw[filled] (76.77112424063552,-155.921324855331) circle (26.248482627285476);
\draw[filled] (133.75448093410674,26.733999290393836) circle (26.655460115935874);
\draw[filled] (-70.83798615579437,-32.73140827745117) circle (27.056316596122045);
\draw[filled] (-127.99058315172994,-27.843584903057348) circle (27.451320214279857);
\draw[filled] (172.15927990758368,-73.4986695185978) circle (27.840720092416266);
\draw[filled] (131.54545832012994,-171.77525183422156) circle (28.22474816577841);
\draw[filled] (28.368679650910128,68.6634761172389) circle (28.603620798405096);
\draw[filled] (-29.214456822140725,70.82258456704106) circle (28.977540208518278);
\draw[filled] (-128.99391564222563,-170.65330426958369) circle (29.3466957304163);
\draw[filled] (71.96477145089877,29.789804109680357) circle (29.711264935222385);
\draw[filled] (-128.02449305346917,29.68110691533845) circle (30.07141462931726);
\draw[filled] (-68.71576800097026,169.57269825361348) circle (30.42730174638653);
\draw[filled] (169.2209258533822,72.51003262745165) circle (30.779074146617827);
\draw[filled] (-31.205410509051475,-129.20305692257935) circle (31.126871334593503);
\draw[filled] (-168.52917489423447,124.6494215053322) circle (31.47082510576549);
\draw[filled] (-168.18893986999146,-71.6150614647661) circle (31.811060130008556);
\draw[filled] (32.069155305121186,-77.84683270065129) circle (32.14769447957911);
\draw[filled] (127.75165394491246,-32.402792474890866) circle (32.48084010781898);
\draw[filled] (120.9450285616189,167.18939671589314) circle (32.81060328410684);%
        \end{tikzpicture}%

    \caption{Example packings of various sets of circles in a square produced by Split Packing.}
    \label{fig:example-circles-in-square}
\end{figure}
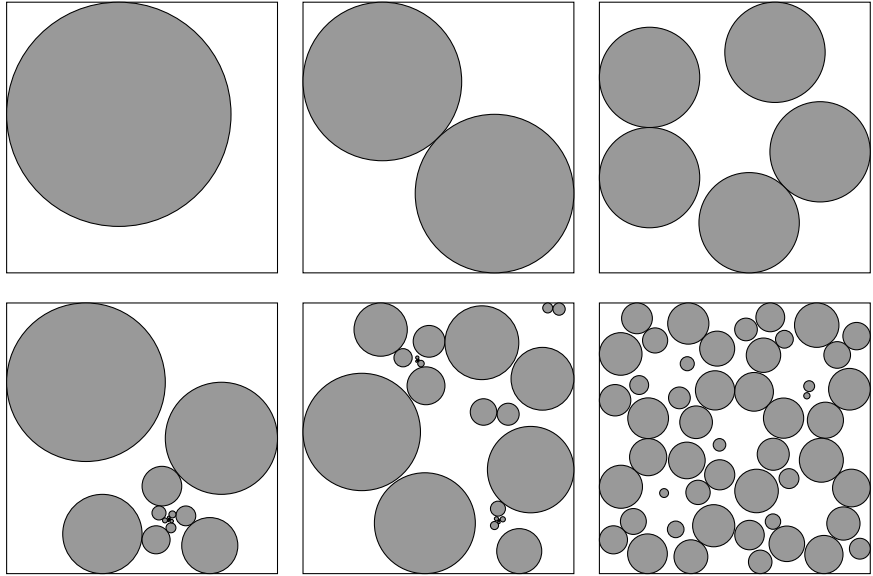

\section{Weighted Greedy Splitting}

The second main result of this paper is an algorithm to pack into not necessarily symmetric, non-acute triangles with critical density.
For this, it is necessary to split the sets of circles recursively into two groups of unequal target area.
In the next sections, we introduce a weighted variant of the Split Packing algorithm.

\Cref{alg:weighted-split} behaves like \cref{alg:split}, except that it splits the sets of circles into two groups according to the \emph{split key} $F$, which determines the targeted ratio of the resulting groups' combined areas.

If we wanted to split $C$ into equally sized halves, we could choose the tuple $(1,1)$.
The tuple $(\frac 1 2,\frac 1 2)$ would give the same result.
For asymmetric containers, we may want to target a different ratio.
For example, if we wanted to make one group three times as large as the other, we could use the tuple $(1,3)$.

In the simplest case, the split key will actually describe the desired areas of the two groups, and \textsc{WeightedSplit} puts the next circle into the group which has the smaller “relative filling level”.

\begin{algorithm}
    \caption{\textsc{WeightedSplit}$(C,F)$}
    \label{alg:weighted-split}
    \begin{algorithmic}
        \Require A set of circles $C$, sorted by size in descending order, and a split key\\$F = (f_1, f_2)$ with $f_i > 0$
        \Ensure Sets of circles $C_1, C_2$
        \State $C_1 \gets \emptyset$
        \State $C_2 \gets \emptyset$
        \ForAll{$c \in C$}
            \State $j = \argmin_{i} \frac{\mysum(C_i)}{f_i}$\Comment{Find the index of the more empty bucket.}
            \State $C_j \gets C_j \cup \{c\}$
        \EndFor
    \end{algorithmic}
\end{algorithm}

If the resulting groups' area ratio deviates from the area ratio targeted by the split key, we gain additional information about the “relatively larger” group: The more this group exceeds its targeted ratio, the larger the minimum size of its elements, allowing a “more rounded” subcontainer in the packing.
See \Cref{fig:hats-in-hat} for an illustration.

\begin{lemma}\label{th:weighted-min1}
    For any $C_1$ and $C_2$ produced by \textsc{WeightedSplit}$(C,(f_1,f_2))$:

    \begin{equation}
        \min(C_i) \ge \mysum(C_i) - f_i\frac{\mysum(C_j)}{f_j}.
    \end{equation}
\end{lemma}

\begin{proof}
    If $\frac{\mysum(C_i)}{f_i} < \frac{\mysum(C_j)}{f_j}$, then the lemma says that $\min(C_i)$ is larger than a negative number, which is certainly true.

    Otherwise, set $r := \frac{\mysum(C_j)}{f_j}$.
    This value describes the smaller “relative filling level” by the time the algorithm ends.
    Now assume for contradiction $C_i$ contained an element smaller than $\mysum(C_i) - f_i r$.
    As the elements were inserted by descending size, all elements which were put into $C_i$ after that element would have to be at least as small.
    So the final element put into $C_i$ (let us call it~$c$) would be smaller than $\mysum(C_i) - f_i r$, as well.

    But this means that $$\frac{\mysum(C_i) - c}{f_i} > \frac{\mysum(C_i) - (\mysum(C_i) - f_i r)}{f_i} = r\text{,}$$ meaning that at the moment before $c$ was inserted, the relative filling level of $C_i$ would already have been larger than~$r$.
    Recall that $r$ is the smallest filling level of any group by the time the algorithm ends, meaning that at the time when $c$ is inserted, $C_i$'s filling level is already larger than the filling level of the other group.
    This is a contradiction, as the greedy algorithm would choose to put $c$ not into $C_i$, but into the other group with the smaller filling level in this case.
\end{proof}

%
%
%
%
%
%

We now define a term that encapsulates all properties of the sets of circles output by \textsc{WeightedSplit}.
These properties depend on the used split key $F$, and also on the combined area $a$ and the minimum circle size $b$ of the set of circles, which is why the term has three parameters.

\begin{definition}\label{def:conjugated}
    For any $0 \le b \le a$ and any split key $F = (f_1, f_2)$, we say that the tuples $(a_1, b_1)$, $(a_2, b_2)$ are \emph{$(a,b,F)$-conjugated} if

    \begin{itemize}
        \item $a_1 + a_2 = a$,
        \item $b_i \ge b$, and
        \item $b_i \ge a_i - f_i \frac{a_j}{f_j}$.
    \end{itemize}

    Two sets of circles $C_1$ and $C_2$ are $(a,b,F)$-conjugated if there are any $(a,b,F)$-conjugated tuples $(a_1, b_1)$ and $(a_2, b_2)$ so that $C_1 \in \C(a_1, b_1)$ and $C_2 \in \C(a_2, b_2)$.
\end{definition}

We can now associate this property with \textsc{WeightedSplit} in the following theorem.

\begin{theorem}\label{th:split-properties}
    For any $C \in \C(a,b)$ and any split key $F = (f_1, f_2)$,\\\textsc{WeightedSplit}$(C,F)$ always produces two $(a,b,F)$-conjugated subsets.
\end{theorem}

\begin{proof}
It follows directly from the algorithm that the combined areas of the subsets add up to $a$.
    As the minimum size of all circles in $C$ is $b$, this must also be true for the subsets, so $\min(C_i) \ge b$.
    The other minimum-size property follows from \Cref{th:weighted-min1}.
\end{proof}

As described, one way to think about conjugatedness is that it gives guarantees for the minimum sizes of the “larger” produced subset.
To provide an intuition of how the conjugatedness property is used in the later sections, we show several examples of shapes with $(a,b,F)$-conjugated parameters in \Cref{fig:conjugated}: $a_1$ and $a_2$ represent the area which can be packed into the respective shape, while $b_1$ and $b_2$ represent their “rounding”.
The shapes can always be packed because if one shape gets larger, it is rounded so much that it still fits inside the container.

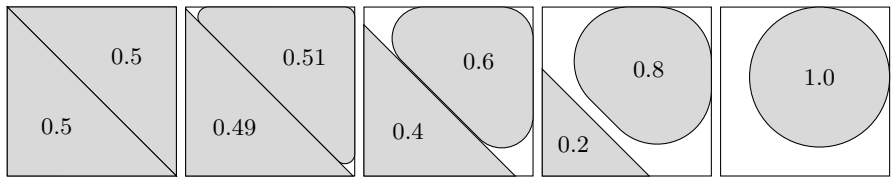
\begin{figure}
    \begin{tikzpicture}[scale=1.65]
        \hatsinsquare{0.5}
    \end{tikzpicture}
    \begin{tikzpicture}[scale=1.65]
        \hatsinsquare{0.49}
    \end{tikzpicture}
    \begin{tikzpicture}[scale=1.65]
        \hatsinsquare{0.4}
    \end{tikzpicture}
    \begin{tikzpicture}[scale=1.65]
        \hatsinsquare{0.2}
    \end{tikzpicture}
    \begin{tikzpicture}[scale=1.65]
        \hatsinsquare{0}
    \end{tikzpicture}

    \caption{The two shapes' parameters are $(a,b,F)$-conjugated, which is why they always can be packed. The numbers represent $a_1$ and $a_2$, and in this case are the areas of the shapes' incircles.}
    \label{fig:conjugated}
\end{figure}

\section{Weighted Split Packing}

We now state a weighted version of the Split Packing theorem, before we apply it to triangular containers in the next section:
If it is possible to find two subcontainers that fit in a given shape, and that can pack all possible subsets produced by \textsc{WeightedSplit} for a fixed split key $F$, it is possible to pack the original class of sets of circles into that shape.

\begin{theorem}[Weighted Split Packing]\label{th:weighted-split-packing}
    A shape $s$ is a $\C(a,b)$-shape if there is a split key $F$, so that for all $(a,b,F)$-conjugated tuples $(a_1, b_1)$ and $(a_2, b_2)$ one can find a $\C(a_1, b_1)$-shape and a $\C(a_2, b_2)$-shape which can be packed into $s$.
\end{theorem}

\begin{proof}
    Consider an arbitrary $C \in \C(a,b)$.
    We use \textsc{WeightedSplit}$(C,F)$ to produce two subsets $C_1$ and $C_2$.
    We know from \Cref{th:split-properties} that those subsets will always be $(a,b,F)$-conjugated.
    So if we can indeed find two shapes which can pack these subsets, and if we can pack these two shapes into $s$, then we also can pack the original set of circles $C$ into $s$.

Note that in the special case that $C$ consists of a single circle, \textsc{WeightedSplit}$(C,F)$ will yield two sets of circles $C_1 = \{C\}$ and $C_2 = \emptyset$.
    For this case, \Cref{th:split-properties} guarantees a minimum size of $a$ for the first group, and the associated $\C(a_1,b_1)$-shape is just an $a$-circle.
    This means that we can simply place the input circle in the container, and stop the recursion at this point.
\end{proof}

Written as an algorithm, Weighted Split Packing looks like this:

\begin{algorithm}
    \caption{\textsc{WeightedSplitpack}$(s,C)$}
    \begin{algorithmic}
        \Require A $\C(a,b)$-shape $s$ and a set of circles $C \in \C(a,b)$, sorted by size in descending order
        \Ensure A packing of $C$ into $s$
        \State Determine split key $F$ for shape $s$
        \State $(C_1, C_2) \gets \textsc{WeightedSplit}(C,F)$ \Comment{See \Cref{alg:weighted-split}.}
        \ForAll{$i \in \{1,2\}$}
            \State $a_i \gets \mysum(C_i)$
            \State $b_i \gets$ minimum guarantee for $C_i$ \Comment{See \Cref{th:weighted-min1}.}
            \State Determine a $\C(a_i, b_i)$-shape $s_i$
            \State \Call{Splitpack}{$s_i, C_i$}
        \EndFor
        \State Pack $s_1$, $s_2$, and their contents into $s$
    \end{algorithmic}
\end{algorithm}

%

%

The analysis of the Weighted Split Packing approach follows exactly the same lines as in the unweighted version, see \cref{sec:analysis}.


\begin{theorem}
    Weighted Split Packing requires $\mathcal{O}(n)$ basic geometric constructions and $\mathcal{O}(n^2)$ numerical operations.
\end{theorem}

%

\begin{theorem}
    Weighted Split Packing, when used to pack circles into a $\C(a,b)$-shape of area $A$, is an approximation algorithm with an approximation factor of $\frac{A}{a}$, compared to the container of minimum area.
\end{theorem}

%
%

\section{Packing into Hats}

As a preparation for the results about asymmetric triangles, we re-introduce (general) \emph{hats}.

\begin{definition}
    For each $0 \le b \le a$, an \emph{$(a,b)$-hat} is a non-acute triangle with an incircle of area $a$, whose corners are rounded to the radius of a $b$-circle, see \Cref{fig:hat}.
    Call the two smaller angles of the original triangle \emph{left-angle} and \emph{right-angle}.
    If we say \emph{right hat}, the hat is based on a right triangle.
\end{definition}

\renewcommand\defaulta{30}
\renewcommand\defaultb{40}
\renewcommand\defaultr{0.2}
\renewcommand\defaultx{0.7}

        \begin{figure}%
            \begin{tikzpicture}[scale=2.5]%
                \pgfmathsetmacro{\a}{(\defaulta)}
\pgfmathsetmacro{\b}{(\defaultb)}
\pgfmathsetmacro{\round}{(\defaultr)}

\hatshape{1}{\round}{\a}{\b}

\draw[dashed] (rightcenter) circle(\s) node {$b$};
\draw[dashed] (leftcenter) circle(\s) node {$b$};
\draw[dashed] (topcenter) circle(\s) node {$b$};
\draw[dashed] (midcenter) circle(\r) node {$a$};

\path (leftcenter) +(-20pt,0) node {\emph{left-angle}};
\path (rightcenter) +(20pt,0) node {\emph{right-angle}};
%
            \end{tikzpicture}%
            \caption{An $(a,b)$-hat.}%
            \label{fig:hat}%
        \end{figure}
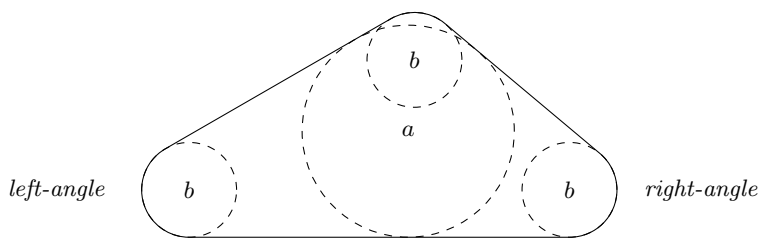%

We now proceed to show that all sets of circles with a combined area of up to $a$ with a minimum circle size of $b$ can be packed into an $(a,b)$-hat.

First, it is important to choose the correct split key when packing into asymmetric hats.
We aim for a group ratio that leads to a cut through the hat's tip if it is reached exactly.

\begin{definition}\label{def:hat-split-key}
    To get a hat's \emph{associated split key}, split the underlying triangle orthogonally to its base through its tip, and inscribe two circles in the two sides, see \Cref{fig:hat-f}.
    The areas of these circles are the two components of the hat's split key.
\end{definition}

        \begin{figure}%
            \begin{tikzpicture}[scale=2.5]%
                \input{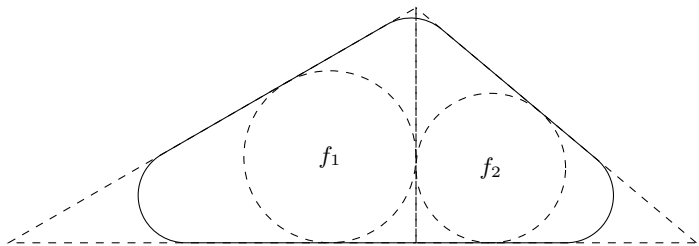}%
            \end{tikzpicture}%
            \caption{A hat's \emph{associated split key} equals $(f_1, f_2)$}%
            \label{fig:hat-f}%
        \end{figure}%

\begin{lemma}\label{th:hats-in-hat}
    Consider an $(a,0)$-hat with the associated split key $F = (f_1, f_2)$, and call its left- and right-angles $\alpha$ and $\beta$.
    For all $(a,0,F)$-conjugated tuples $(a_1, b_1)$ and $(a_2, b_2)$, the following two shapes can be packed into the hat.
    \begin{itemize}
        \item a right $(a_1,b_1)$-hat with a right-angle of $\alpha$ and
        \item a right $(a_2,b_2)$-hat with a left-angle of $\beta$.
    \end{itemize}
\end{lemma}

The proof of this theorem is rather technical.
See \Cref{fig:hats-in-hat} for an intuition of what the resulting hats look like.
Note that, as the hats' incircles are getting larger than the targeted area ratio, their corners become more rounded so that they don't overlap the container's boundary.

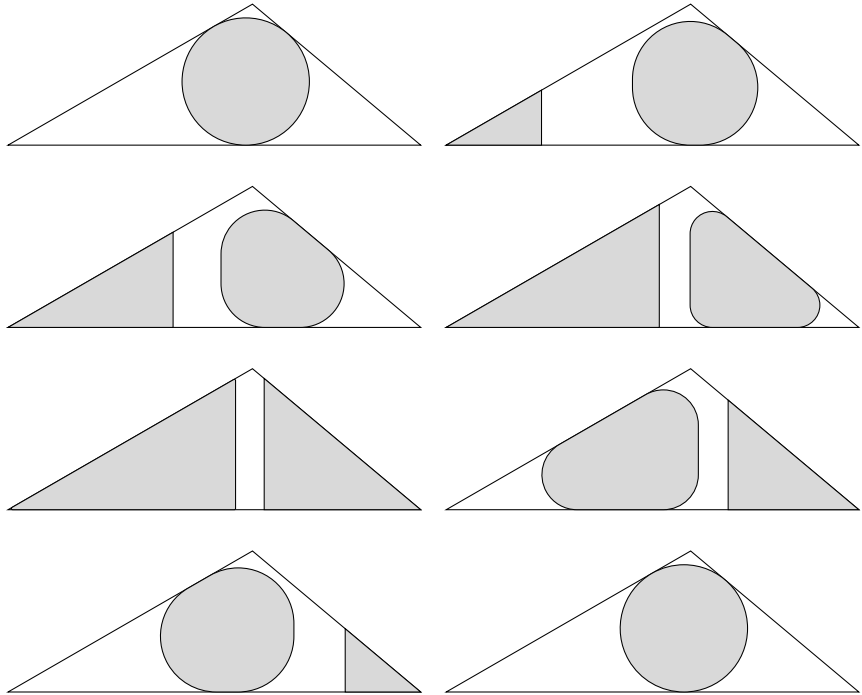
\begin{figure}
        \begin{tikzpicture}[scale=1.5,baseline={([yshift={-\ht\strutbox}]current bounding box.north)},outer sep=0pt,inner sep=0pt]
            \hatsinhatnotext{\defaulta}{\defaultb}{1}{0}
        \end{tikzpicture}
        ~
        \begin{tikzpicture}[scale=1.5,baseline={([yshift={-\ht\strutbox}]current bounding box.north)},outer sep=0pt,inner sep=0pt]
            \hatsinhatnotext{\defaulta}{\defaultb}{0.9}{0}
        \end{tikzpicture}

        \vspace{5mm}

        \begin{tikzpicture}[scale=1.5,baseline={([yshift={-\ht\strutbox}]current bounding box.north)},outer sep=0pt,inner sep=0pt]
            \hatsinhatnotext{\defaulta}{\defaultb}{0.7}{0}
        \end{tikzpicture}
        ~
        \begin{tikzpicture}[scale=1.5,baseline={([yshift={-\ht\strutbox}]current bounding box.north)},outer sep=0pt,inner sep=0pt]
            \hatsinhatnotext{\defaulta}{\defaultb}{0.5}{0}
        \end{tikzpicture}

        \vspace{5mm}

        \begin{tikzpicture}[scale=1.5,baseline={([yshift={-\ht\strutbox}]current bounding box.north)},outer sep=0pt,inner sep=0pt]
            \hatsinhatnotext{\defaulta}{\defaultb}{0.43}{0}
        \end{tikzpicture}
        ~
        \begin{tikzpicture}[scale=1.5,baseline={([yshift={-\ht\strutbox}]current bounding box.north)},outer sep=0pt,inner sep=0pt]
            \hatsinhatnotext{\defaulta}{\defaultb}{0.3}{0}
        \end{tikzpicture}

        \vspace{5mm}

        \begin{tikzpicture}[scale=1.5,baseline={([yshift={-\ht\strutbox}]current bounding box.north)},outer sep=0pt,inner sep=0pt]
            \hatsinhatnotext{\defaulta}{\defaultb}{0.1}{0}
        \end{tikzpicture}
        ~
        \begin{tikzpicture}[scale=1.5,baseline={([yshift={-\ht\strutbox}]current bounding box.north)},outer sep=0pt,inner sep=0pt]
            \hatsinhatnotext{\defaulta}{\defaultb}{0}{0}
        \end{tikzpicture}
    \caption{Hat-in-hat packings for different ratios of $a_1$ and $a_2$.}
    \label{fig:hats-in-hat}
\end{figure}

As a preparation for the proof, we establish the following lemma.

\begin{lemma}\label{th:overlap}
    Place two circles of combined area $a$ in two corners of a triangle, like in \Cref{fig:hats-overlap}.
    Let $w$ be the length of the connecting side of the triangle.
    Now define $p_1$ and $p_2$ to be the “projection factors”, so that, when projecting circle $a_i$ down onto the connecting side, the distance between the triangle's corner and the far point of the projection is $\sqrt{a_i}p_i$.
    The two projections do not intersect if

    \begin{equation}
        w \ge \sqrt{a(p_1^2 + p_2^2)}.
    \end{equation}
\end{lemma}

\begin{proof}
    Let $w'(a_1) = \sqrt{a_1}p_1 + \sqrt{a-a_1}p_2$ be the combined width of both projections.
    This function has its global maximum at $a_1 = \frac{p_1^2}{p_1+p_2}a$, and the maximum value is $\sqrt{a(p_1^2+p_2^2)}$.
    If $w$ is at least as large as this value, the two projections do never intersect.
\end{proof}

        \begin{figure}%
            \begin{tikzpicture}[scale=2.5]%
                \input{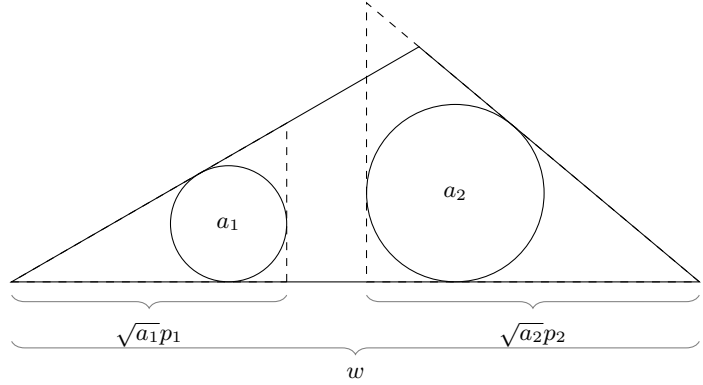}%
            \end{tikzpicture}%
            \caption{The circles' projections do not overlap if $w \ge \sqrt{a(p_1^2 + p_2^2)}$.}%
            \label{fig:hats-overlap}%
        \end{figure}%

We can now proceed to prove \Cref{th:hats-in-hat}.

\begin{proof}
    Place the hats' tips at the bottom of the container hat, rotate their $\alpha$- and $\beta$-angles toward the container's matching angles and push them as far to the left/right as possible.
    \Cref{fig:hats-in-hat} illustrates how these packings look like for different values of $a_1$ and~$a_2$.

    This way of placing the two hats results in a valid packing if (1) the hats do not overlap each other and (2) the hats fit into the hat individually.
    We are going to prove these two properties separately.

    We first want to show that the hats do not overlap each other.
    If the hats' projections onto the container's base do not overlap, we found a separating axis and can be sure that the hats do not overlap, as well.
    Furthermore, because the hats' incircles touch the rightmost part of the left hat's boundary and the leftmost part of the right hat's boundary, it suffices to show that the projections of the hats' incircles onto the container's base do not overlap.

    We want to use \Cref{th:overlap} for this proof, so we need to make a statement about the projection factors $p_1$ and $p_2$ in \Cref{fig:hatlr}: If the top angle is a right angle, we can see that $\sqrt{a}p_1 = x$ and $\sqrt{a}p_2 = y$.
    So by the Pythagorean theorem, $w^2 = (\sqrt{a}p_1)^2 + (\sqrt{a}p_2)^2$.
     If the top angle is more obtuse, but the incircle's center stays at the same $x$-coordinate (like the dotted variant in \Cref{fig:hatlr}),
    both $\sqrt{a}p_1$ and $\sqrt{a}p_2$ only get smaller, so for each hat, $w^2 \ge (\sqrt{a}p_1)^2 + (\sqrt{a}p_2)^2$, which is equivalent to $w \ge \sqrt{a(p_1^2+p_2^2)}$.
    By \Cref{th:overlap}, this means that the projections of the circles do not overlap, which in turn means that the two hats do not overlap.

        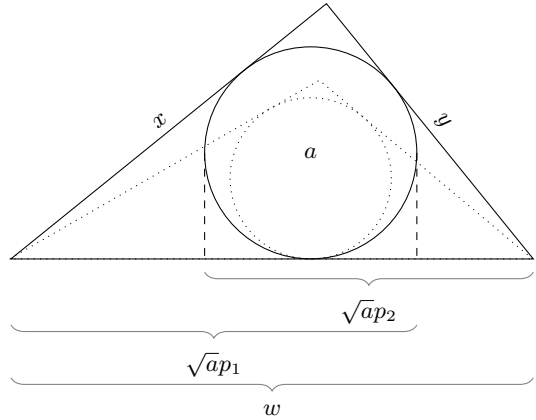
\begin{figure}%
            \begin{tikzpicture}[scale=2.5]%
                \pgfmathsetmacro{\aa}{(\defaulta)}
\pgfmathsetmacro{\bb}{(\defaultb)}
\pgfmathsetmacro{\a}{39}
\pgfmathsetmacro{\b}{51}

\hatshape{1}{0}{\a}{\b}
\draw (midcenter) circle(\r) node {$a$};
\draw[draw=none] (left) -- node[sloped,above] {$x$} (top);
\draw[draw=none] (top) -- node[sloped,above] {$y$} (right);

\coordinate (lc) at ($(midcenter)-(\r,0)$);
\coordinate (rc) at ($(midcenter)+(\r,0)$);
\coordinate (l) at ($(midcenter)+(-\r,-\r)$);
\coordinate (r) at ($(midcenter)+(\r,-\r)$);

\draw[dashed] (lc) -- (l);
\draw[dashed] (rc) -- (r);

\draw[bracket] ([yshift=-2pt]right) -- node[below] {$\sqrt{a}p_2$} ([yshift=-2pt]l);
\draw[bracket] ([yshift=-10pt]r) -- node[below] {$\sqrt{a}p_1$} ([yshift=-10pt]left);
\draw[bracket] ([yshift=-18pt]right) -- node[below] {$w$} ([yshift=-18pt]left);

\begin{scope}[scale=0.76,shift={(0,-0.178)}]
    \hatshape[dotted]{1}{0}{\aa}{\bb}
    \draw[dotted] (midcenter) circle(\r);
\end{scope}%
            \end{tikzpicture}%
            \caption{$w \ge \sqrt{a(p_1^2+p_2^2)}$ holds for each non-acute triangle.}%
            \label{fig:hatlr}%
        \end{figure}%

    The second property we need to show is that the hats fit into the container individually.
    Unfortunately, this part of the proof is going to be long and technical.

    If a hat's incircle is not larger than the incircle of the container hat's side, it clearly fits into the container because it is a subset of that side (like all the non-rounded hats in \Cref{fig:hats-in-hat}).
    So let us assume $a_i > f_i$.

    In this proof, we are going to use two different length-area ratios, which are illustrated in \Cref{fig:hat-poke-f}.
    The first one is $d$, which describes the ratio between the length of the triangle's right leg and the square root of the area of its right incircle $f_i$.
    Note that for all triangles \emph{similar} to the right part of the container triangle, this ratio between the length of this edge and the square root of the incircle's area is a constant.
    The second ratio, $e$, is the ratio between the length of the same right leg and the square root of the incircle-area of the whole container triangle.
    Again, it is a constant for triangles similar to the given container triangle.
    Note that, in preparation for a generalization later in this section, we denote the triangle's incircle by $o$.
    From \Cref{fig:hat-poke-f} we can now observe that
    $e\sqrt{o} = d\sqrt{f_i}$, which is equivalent to $e = d\sqrt{f_i/o}$.

        \begin{figure}%
            \begin{tikzpicture}[scale=2.5]%
                \input{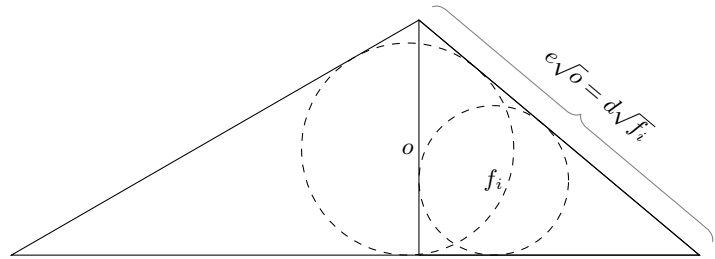}%
            \end{tikzpicture}%
            \caption{The ratios $d$ and $e$ are constant for all similar triangles.}%
            \label{fig:hat-poke-f}%
        \end{figure}%

    Moving forward, in \Cref{fig:hat-poke}, we display the situation when packing a hat into (\wlofg) the right leg of the container.
    $f_i$ is the relevant factor from the split key, $a_i$ is the hat's incircle and $b_i$ represents the hat's rounding.

        \begin{figure}%
            \begin{tikzpicture}[scale=2.5]%
                \input{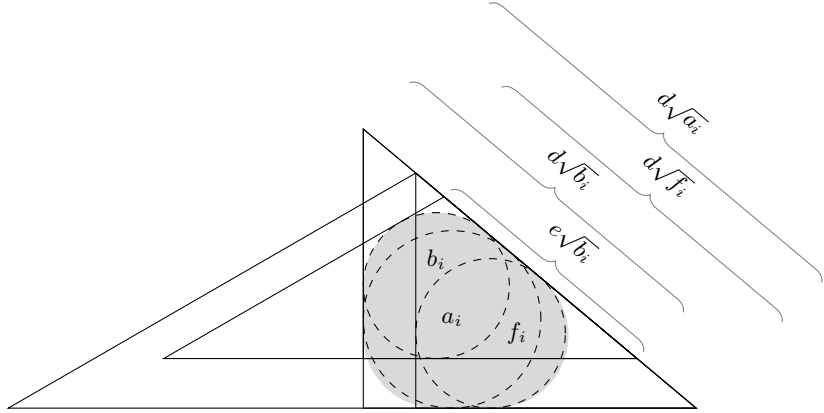}%
            \end{tikzpicture}%
            \caption{Various measurements when packing a rounded hat.}%
            \label{fig:hat-poke}%
        \end{figure}%

    The hat is placed in such a way that it will never overlap the bottom or the right leg of the containing triangle, so it is sufficient to show that it does not overlap the left leg.
    We can tell from \Cref{fig:hat-poke} that this does not happen if the length of the right side of the triangle the hat is based on ($d\sqrt{a_i}$), minus the length of the right side of the $(b_i,0)$-triangle similar to the containers right side ($d\sqrt{b_i}$), plus the length of the right side of the $(b_i,0)$-triangle similar to the container ($e\sqrt{b_i}$) is at most the length of the container's right leg ($d\sqrt{f_i}$).
    So the following condition has to hold:
    \begin{equation*}
        d\sqrt{a_i} - d\sqrt{b_i} + e\sqrt{b_i} \le d\sqrt{f_i}.
    \end{equation*}

    As previously observed, $e = d\sqrt{f_i/o}$:
    \begin{equation*}
        d\sqrt{a_i} - d\sqrt{b_i} + d\sqrt{f_i/o}\sqrt{b_i} \le d\sqrt{f_i}.
    \end{equation*}

    In our case, the incircle of the triangle has exactly the maximal area which we want to pack, so $o = a$.
    But even if $o \ge a$, the inequality is true if
    \begin{equation*}
        d\sqrt{a_i} - d\sqrt{b_i} + d\sqrt{f_i/a}\sqrt{b_i} \le d\sqrt{f_i}.
    \end{equation*}

    We can also divide by $d$ and factor out $\sqrt{b_i}$ to get the following:
    \begin{equation}\label{eq:tripoke}
        \sqrt{a_i} - (1-\sqrt{f_i/a})\sqrt{b_i} \le \sqrt{f_i}.
    \end{equation}

    Let $j$ be the index of the other hat to be packed.
    We know (from the conjugatedness) that the sum of both hats' incircles does not exceed the total area $a$, so $a_i + a_j \le a$.
    Also, $f_i + f_j \ge a$, as demonstrated in \Cref{fig:hat-f-sum}: In right triangles, $f_1 + f_2$ is exactly $a$, because as its two halves are \emph{similar} to the large triangle, the two halves' areas add up to the container triangle's area, and the ratio between the areas of a triangle and its incircle are constant.
    When making the upper angle more obtuse, but letting $f_1$ and $f_2$ stay the same, the incircle only shrinks (like the dotted variant in \Cref{fig:hat-f-sum}).

        \begin{figure}%
            \begin{tikzpicture}[scale=2.5]%
                \input{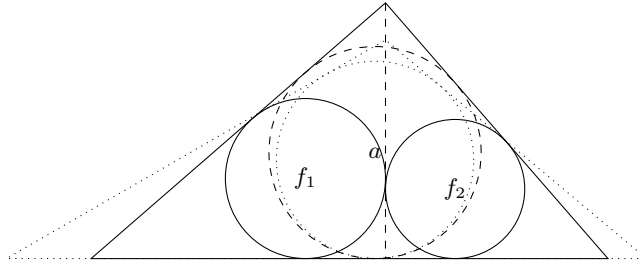}%
            \end{tikzpicture}%
            \caption{In non-acute triangles, $f_1 + f_2 \ge a$.}%
            \label{fig:hat-f-sum}%
        \end{figure}%

    Putting it together, by \Cref{th:split-properties}, our hat is rounded by
    $$b_i \ge a_i - f_i\frac{a_j}{f_j} \ge a_i - f_i\frac{a-a_i}{a-f_i} = \frac{a_i(a-f_i)-f_i(a-a_i)}{a-f_i} = a\frac{a_i-f_i}{a-f_i}.$$

    Inserting this into \Cref{eq:tripoke} yields
    $$\sqrt{a_i} - (1-\sqrt{f_i/a})\sqrt{a\frac{a_i-f_i}{a-f_i}} \le \sqrt{f_i}.$$

    Bringing the subtrahend to the right and squaring both sides (both are positive) yields
    $$a_i \le f_i + 2\sqrt{f_i}(1-\sqrt{f_i/a})\sqrt{a\frac{a_i-f_i}{a-f_i}} + (1-\sqrt{f_i/a})^2a\frac{(a_i-f_i)}{a-f_i}.$$

    Subtracting $f_i$ and dividing by $\sqrt{a_i-f_i}$ results in
    $$\sqrt{a_i-f_i} \le \frac{2\sqrt{f_i}(1-\sqrt{f_i/a})\sqrt{a}}{\sqrt{a-f_i}} + (1-\sqrt{f_i/a})^2a\frac{\sqrt{a_i-f_i}}{a-f_i}.$$

    After rearranging, we get
    $$\sqrt{a_i-f_i}\frac{(a-f_i)-(1-\sqrt{f_i/a})^2a}{a-f_i} \le \frac{2\sqrt{f_i}(1-\sqrt{f_i/a})\sqrt{a}}{\sqrt{a-f_i}},$$

    which simplifies to
    $$\sqrt{a_i-f_i}\frac{2\sqrt{f_ia}-2f_i}{a-f_i} \le \frac{2\sqrt{f_ia}-2f_i}{\sqrt{a-f_i}}.$$

    Multiplying with $\sqrt{a-f_i}$ yields
    $$\sqrt{a_i-f_i}\frac{2\sqrt{f_ia}-2f_i}{\sqrt{a-f_i}} \le \frac{2\sqrt{f_ia}-2f_i}{\sqrt{a-f_i}}\sqrt{a-f_i}.$$

    Finally, divide by the fraction to get
    \begin{equation*}
        \sqrt{a_i-f_i} \le \sqrt{a-f_i} \iff a_i - f_i \le a - f_i \iff a_i \le a.
    \end{equation*}

    From the conjugatedness we know that $a_i$ is less or equal than $a$, so \Cref{eq:tripoke} is true and the hat always fits into the container.
    This completes the proof of \Cref{th:hats-in-hat}.
\end{proof}

\label{endofproof} In the previous lemma, the container is always an $(a,0)$-hat, which is essentially a non-rounded triangle with an incircle of $a$.
The next lemma extends this idea to hats which are actually rounded.
It is identical to \Cref{th:hats-in-hat}, except that the rounding of the container hat is no longer 0, but $b$.

\begin{lemma}\label{th:rounded-hats-in-hat}
    Consider an $(a,b)$-hat with the associated split key $F = (f_1, f_2)$, and call its left- and right-angles $\alpha$ and $\beta$.
    For all $(a,b,F)$-conjugated tuples $(a_1, b_1)$ and $(a_2, b_2)$ with $a_1 + a_2 \le a$, the following two shapes can be packed into the hat.
    \begin{itemize}
        \item a right $(a_1,b_1)$-hat with a right-angle of $\alpha$ and
        \item a right $(a_2,b_2)$-hat with a left-angle of $\beta$.
    \end{itemize}
\end{lemma}

\begin{proof}
    \Cref{th:hats-in-hat} tells us that \Cref{th:rounded-hats-in-hat} is true for $b = 0$.
    Now the container's corners can be rounded to the radius of a $b$-circle, and we need to show that the two hats from the previous construction still fit inside.
    But all of the two hat's corners are also rounded to (at least) the same radius (see \Cref{th:split-properties}), so they will never overlap the container, see \Cref{fig:hats-rounding}.
\end{proof}

        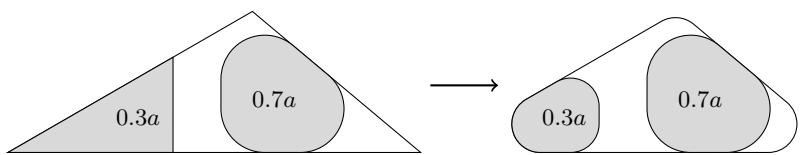
\begin{figure}%
            \begin{tikzpicture}[scale=1.5]%
                    \hatsinhat{\defaulta}{\defaultb}{\defaultx}{0}
\end{tikzpicture}
\begin{tikzpicture}[scale=1.5]
\draw[->,thick] (0,0) -- (0.6,0);
\draw[draw=none] (0,0) -- (0,-0.6);
\end{tikzpicture}
\begin{tikzpicture}[scale=1.5]
    \hatsinhat{\defaulta}{\defaultb}{\defaultx}{\defaultr}%
            \end{tikzpicture}%
            \caption{Rounding all hats' corners by the same radius does not affect the packing.}%
            \label{fig:hats-rounding}%
        \end{figure}%

With these preparations, we can apply Split Packing to (general) hats.

\begin{theorem}\label{th:hats}
    Given an $(a,b)$-hat, all sets of circles with a combined area of at most $a$ and a minimum circle size of at least $b$ can be packed into that hat.
\end{theorem}

\begin{proof}
    We prove by induction that we can pack each $C \in \C(a,b)$ into the hat.

    If $C$ only consists of a single circle, it can be packed into the hat, as it is at most as big as the hat's incircle.

    Now assume that for any $0 \le b \le a$, any $(a,b)$-hat could pack all sets of circles into $\C(a,b)$ with at most $n$ circles.
    Consider a set of circles $C \in \C(a,b)$ containing $n+1$ circles.
    \Cref{def:hat-split-key} tells us how to compute the split key $F$.
    Then we know from \Cref{th:split-properties} that \textsc{Split} will partition $C$ into two subsets $C_1 \in \C(a_1, b_1)$ and $C_2 \in \C(a_2, b_2)$, whose parameters are $(a,b,F)$-conjugated.
    As \textsc{Split} can never return an empty set (except for $|C| = 1$, a case which we handled above), each subset will contain at most $n$ circles.
    We know from \Cref{th:rounded-hats-in-hat} that, for all pairs of $(a,b,F)$-conjugated tuples, we can find two hats with matching parameters which fit into the container hat.
    By assumption, these hats can now pack all sets from $\C(a_1, b_1)$ and $\C(a_2, b_2)$, respectively, which means that they can especially also pack $C_1$ and $C_2$.
    If we then pack the two hats into the container, we have constructed a packing of $C$ into the container hat.

    By induction, we can pack each $C \in \C(a,b)$ into the $(a,b)$-hat.
\end{proof}

Finally, we can state this paper's second central result.

\begin{theorem}\label{th:tri}
    Given a non-acute triangle with an incircle of area $a$, all sets of circles with a combined area of up to $a$ can be packed into the triangle, and this bound is tight.
    See \Cref{fig:example-circles-in-right} for some example packings.
    Expressed algebraically, for a triangle with side lengths $x$, $y$, and $z$, the critical density is

    \begin{equation}
        \phi_t = \pi \sqrt{\dfrac{(x+y-z)(z+x-y)(y+z-x)}{(x+y+z)^3}}.
    \end{equation}
\end{theorem}

\begin{proof}
    The triangle is an $(a,0)$-hat, which by \Cref{th:hats} is a $\C(a)$-shape.

    On the other hand, a single circle of area $a + \varepsilon$ cannot be packed, as the incircle is by definition the largest circle which fits into the triangle.

    As for the algebraic formulation of the critical density, the area of the triangle can be calculated using Heron's formula:

    $$\Delta(x,y,z) := \sqrt{s(s-x)(s-y)(s-z)} \text{ with } s = \frac{x+y+z}{2}.$$

    It is also known that the radius of the incircle of this triangle is

    $$R(x,y,z) := \frac{\Delta(x,y,z)}{s} \text{ with } s = \frac{x+y+z}{2},$$

    so the incircle has an area of

    $$I(x,y,z) = \pi R(x,y,z)^2 = \pi \frac{(x+y-z)(z+x-y)(y+z-x)}{4(x+y+z)}.$$

    Finally, the ratio between the areas of the circle and the triangle can be calculated to be

    $$\frac{I(x,y,z)}{\Delta(x,y,z)} = \pi \sqrt{\dfrac{(x+y-z)(z+x-y)(y+z-x)}{(x+y+z)^3}}\text{,}$$

    For a right isosceles triangle, this density is approximately 53.90\%.
\end{proof}

\begin{figure}[H]
    \begin{tabular}{cc}
        \begin{tikzpicture}[scale=0.0115]%
            \draw (-311,-200) -- (204,-200) arc (-90:38.99099404250547:0) (204,-200) -- (0,52) arc (38.99099404250548:129.0174830741425:0) -- (0,52) -- (-311,-200) arc (129.01748307414255:270.00000000000006:0);
\draw[filled] (-15.470497437021207,-95.29676904693574) circle (104.70323095306425);%
        \end{tikzpicture}%
     &
        \begin{tikzpicture}[scale=0.0115]%
            \draw (-311,-200) -- (204,-200) arc (-90:38.99099404250547:0) (204,-200) -- (0,52) arc (38.99099404250548:129.0174830741425:0) -- (0,52) -- (-311,-200) arc (129.01748307414255:270.00000000000006:0);
\draw[filled] (-102.02908469703046,-125.96363538094703) circle (74.03636461905295);
\draw[filled] (48.810922991897435,-125.96363538094698) circle (74.03636461905295);%
        \end{tikzpicture}%
     \\
        \begin{tikzpicture}[scale=0.0115]%
            \draw (-311,-200) -- (204,-200) arc (-90:38.99099404250547:0) (204,-200) -- (0,52) arc (38.99099404250548:129.0174830741425:0) -- (0,52) -- (-311,-200) arc (129.01748307414255:270.00000000000006:0);
\draw[filled] (-47.80576607204117,-47.00354090915053) circle (46.82470837498235);
\draw[filled] (105.84980973502546,-153.17529162501765) circle (46.82470837498235);
\draw[filled] (-178.83518858249627,-153.17529162501762) circle (46.82470837498235);
\draw[filled] (45.799240514690425,-78.9951767057803) circle (46.82470837498235);
\draw[filled] (-83.41139705254764,-153.17529162501765) circle (46.82470837498235);%
        \end{tikzpicture}%
     &
        \begin{tikzpicture}[scale=0.0115]%
            \draw (-311,-200) -- (204,-200) arc (-90:38.99099404250547:0) (204,-200) -- (0,52) arc (38.99099404250548:129.0174830741425:0) -- (0,52) -- (-311,-200) arc (129.01748307414255:270.00000000000006:0);
\draw[filled] (-74.11049026943945,-103.3414945585903) circle (74.03636461905297);
\draw[filled] (94.26475128148935,-147.64838452346785) circle (52.351615476532146);
\draw[filled] (36.895693738972454,-52.41092810832633) circle (37.01818230952649);
\draw[filled] (-237.1176243592553,-173.82419226173394) circle (26.175807738266073);
\draw[filled] (-184.08303840010288,-181.49090884523676) circle (18.509091154763244);
\draw[filled] (12.965415298579023,-117.60537279399287) circle (13.087903869133036);
\draw[filled] (-201.68726987380413,-145.06497175873298) circle (9.254545577381622);
\draw[filled] (40.08948590704475,-104.06721301563775) circle (6.543951934566518);
\draw[filled] (28.834871144622955,-102.15053386976206) circle (4.627272788690811);
\draw[filled] (-207.66983948390248,-162.98688578025295) circle (3.271975967283259);
\draw[filled] (-200.656657635136,-158.53763639632933) circle (2.3136363943454055);
\draw[filled] (-205.92042469509403,-157.15304608493287) circle (1.6359879836416296);
\draw[filled] (-204.12711239559096,-160.12734274719443) circle (1.1568181971727027);
\draw[filled] (-195.5626846948366,-156.33505209311207) circle (0.8179939918208148);
\draw[filled] (-197.2200155060374,-156.09546719987762) circle (0.5784090985863514);
\draw[filled] (-203.37929119432872,-158.08916243282104) circle (0.4089969959104074);
\draw[filled] (-196.669883272507,-157.2337777338028) circle (0.2892045492931757);
\draw[filled] (-204.22634782708516,-158.5123179867631) circle (0.2044984979552037);
\draw[filled] (-203.8746043949524,-158.57221421007176) circle (0.14460227464658784);
\draw[filled] (-196.48292797219142,-156.67371792070432) circle (0.10224924897760185);
\draw[filled] (-196.70208990507086,-156.8127569640374) circle (0.07230113732329392);
\draw[filled] (-196.53759718444715,-156.85602541107886) circle (0.051124624488800924);
\draw[filled] (-196.59363819380658,-156.76307864038316) circle (0.03615056866164696);
\draw[filled] (-196.86127655922118,-156.8815877233233) circle (0.025562312244400462);
\draw[filled] (-196.8094849720985,-156.88907475123682) circle (0.01807528433082348);
\draw[filled] (-196.61700760634605,-156.82677177520736) circle (0.012781156122200231);
\draw[filled] (-196.82667660366886,-156.8535025479495) circle (0.00903764216541174);
\draw[filled] (-196.5905370865724,-156.81354816414665) circle (0.0063905780611001155);
\draw[filled] (-196.60152906882655,-156.81167640716828) circle (0.00451882108270587);
\draw[filled] (-196.83251895680374,-156.87100441474064) circle (0.0031952890305500577);
\draw[filled] (-196.82567014302586,-156.86665944190148) circle (0.002259410541352935);
\draw[filled] (-196.8308105405462,-156.8653073090002) circle (0.0015976445152750289);
\draw[filled] (-196.82905925900374,-156.8682118955844) circle (0.0011297052706764675);
\draw[filled] (-196.82069556757338,-156.86450848674258) circle (0.0007988222576375144);
\draw[filled] (-196.82231403139477,-156.86427451712024) circle (0.0005648526353382338);
\draw[filled] (-196.82832896486187,-156.86622148512117) circle (0.0003994111288187572);
\draw[filled] (-196.82177681618785,-156.86538611974294) circle (0.0002824263176691169);
\draw[filled] (-196.8291561686048,-156.86663472296684) circle (0.0001997055644093786);
\draw[filled] (-196.82881266915933,-156.8666932153724) circle (0.00014121315883455844);
\draw[filled] (-196.82159424265237,-156.86483926218907) circle (0.0000998527822046893);
\draw[filled] (-196.82180837617742,-156.86497513017807) circle (0.00007060657941727922);
\draw[filled] (-196.82164773868823,-156.86501719005028) circle (0.00004992639110234465);
\draw[filled] (-196.82170246623645,-156.86492642171956) circle (0.00003530328970863961);
\draw[filled] (-196.8219635918929,-156.86504215324584) circle (0.000024963195551172326);
\draw[filled] (-196.8219137615005,-156.86504946479653) circle (0.000017651644854319805);
\draw[filled] (-196.82172528792836,-156.86498862204652) circle (0.000012481597775586163);
\draw[filled] (-196.82192980003904,-156.86501565211455) circle (0.000008825822427159902);
\draw[filled] (-196.8216994378114,-156.86497570836383) circle (0.0000062407988877930815);
\draw[filled] (-196.82171017216902,-156.86497388047616) circle (0.000004412911213579951);
\draw[filled] (-196.82193550546205,-156.86503025248265) circle (0.0000031203994438965408);%
        \end{tikzpicture}%
     \\
        \begin{tikzpicture}[scale=0.0115]%
            \draw (-311,-200) -- (204,-200) arc (-90:38.99099404250547:0) (204,-200) -- (0,52) arc (38.99099404250548:129.0174830741425:0) -- (0,52) -- (-311,-200) arc (129.01748307414255:270.00000000000006:0);
\draw[filled] (-300.82719132858534,-197.12856636489576) circle (0);
\draw[filled] (-300.82719133037136,-197.12856635985486) circle (0);
\draw[filled] (-300.82719125510914,-197.12856639563842) circle (0);
\draw[filled] (-300.8271920748867,-197.12856668985984) circle (0);
\draw[filled] (-300.82719066689936,-197.1285635861154) circle (0.0000017330674021492615);
\draw[filled] (-300.82718490072114,-197.12857875484133) circle (0.000006802597872294359);
\draw[filled] (-300.8272214490032,-197.1285639414198) circle (0.000021616019384226457);
\draw[filled] (-300.82710016942036,-197.1285135874537) circle (0.00005884577995594496);
\draw[filled] (-300.82718367422063,-197.12875441608645) circle (0.00014235058027014887);
\draw[filled] (-300.82746813797314,-197.12813114759672) circle (0.00031371871035215264);
\draw[filled] (-300.82634138733073,-197.1284586171431) circle (0.0006411882567274131);
\draw[filled] (-300.8287734343278,-197.12966962841895) circle (0.0012314017498072374);
\draw[filled] (-300.82776035698987,-197.1254892097574) circle (0.0022444790876325047);
\draw[filled] (-300.7942444069893,-197.23481842628817) circle (0.003912917475544675);
\draw[filled] (-300.82183154546595,-197.139668669157) circle (0.006564850855305603);
\draw[filled] (-300.84093556370993,-197.13558129599926) circle (0.010652224013071384);
\draw[filled] (-300.80711594788835,-197.11607028198836) circle (0.016784458374643187);
\draw[filled] (-300.81701464767013,-197.172001737759) circle (0.02576820751740473);
\draw[filled] (15.137456141203893,17.859263809256674) circle (0.03865382073317938);
\draw[filled] (15.040852366433018,17.877399141157124) circle (0.056789152633627);
\draw[filled] (-300.8926605687209,-196.98951393987304) circle (0.0818813758581213);
\draw[filled] (15.280215781923902,17.99490228673838) circle (0.11606747247337851);
\draw[filled] (-300.61349342039114,-197.06962666170267) circle (0.16199409768779732);
\draw[filled] (15.17337572731794,17.585797218197218) circle (0.22290752707934156);
\draw[filled] (-301.20892598601824,-197.3709298913074) circle (0.302754415687087);
\draw[filled] (14.753574375970038,18.395836107091665) circle (0.40629411407266286);
\draw[filled] (-300.97245709886045,-196.35277018659386) circle (0.5392233028448887);
\draw[filled] (16.118900665533946,18.093816497981994) circle (0.7083137231823325);
\draw[filled] (-91.01322753580568,-138.40417790786475) circle (0.9215637966052919);
\draw[filled] (-299.8155606854972,-198.81163505734017) circle (1.188364942659814);
\draw[filled] (-104.80758511831279,-146.17907898340036) circle (1.5196834182989756);
\draw[filled] (41.439811874730104,-86.84761516582275) circle (1.9282585175964044);
\draw[filled] (-304.14455376023494,-197.57118201498182) circle (2.4288179850181826);
\draw[filled] (-107.43773446427537,-140.52092562108436) circle (3.038311509824905);
\draw[filled] (-99.49409953815857,-141.6020474123607) circle (3.7761631832877733);
\draw[filled] (4.64537376142343,13.895792379069986) circle (4.664543814292238);
\draw[filled] (-236.80507158531597,-173.89878097026775) circle (5.728664012580072);
\draw[filled] (-235.5366466355408,-161.17245642943973) circle (6.997088962355278);
\draw[filled] (8.48290576939096,28.008561161764117) circle (8.502075822259771);
\draw[filled] (-122.10411275461452,-151.5518422046961) circle (10.279934700819338);
\draw[filled] (-221.01669959904942,-187.62858583062334) circle (12.37141416937668);
\draw[filled] (69.08147484688193,-99.74146893549982) circle (14.82211228727348);
\draw[filled] (-260.25901825179744,-181.64440457394286) circle (17.68291412662385);
\draw[filled] (-111.37359065898995,-114.56693372797555) circle (21.010456796443872);
\draw[filled] (24.848452925302585,-122.62464130639187) circle (24.867622978171354);
\draw[filled] (-76.60931955712655,-170.67593600326722) circle (29.3240639967328);
\draw[filled] (38.61398082683782,-50.462577072150694) circle (34.45675346329513);
\draw[filled] (-172.7346200423787,-139.8993500879107) circle (40.350572537683554);
\draw[filled] (105.27501286901293,-152.90107212984165) circle (47.098927870158356);
\draw[filled] (-54.77518152419314,-62.92127133430252) circle (54.804403293826574);%
        \end{tikzpicture}%
     &
        \begin{tikzpicture}[scale=0.0115]%
            \draw (-311,-200) -- (204,-200) arc (-90:38.99099404250547:0) (204,-200) -- (0,52) arc (38.99099404250548:129.0174830741425:0) -- (0,52) -- (-311,-200) arc (129.01748307414255:270.00000000000006:0);
\draw[filled] (-183.656051413036,-170.70926795100266) circle (2.9322769807392186);
\draw[filled] (-43.82042356212926,-152.39353736398135) circle (4.146865874795834);
\draw[filled] (-192.71823419564234,-175.2894488390531) circle (5.078852712504994);
\draw[filled] (33.27676526679852,-44.98408504696542) circle (5.864553961478437);
\draw[filled] (-116.55013084041532,-103.39231537881595) circle (6.556770657790735);
\draw[filled] (85.25461774777995,-64.72998013900603) circle (7.182582387319944);
\draw[filled] (-84.40101364898717,-83.40190246563873) circle (7.758075666195308);
\draw[filled] (-173.19828889030404,-99.01541251692936) circle (8.293731749591668);
\draw[filled] (131.59464564241324,-170.1983327670557) circle (8.796830942217658);
\draw[filled] (-45.56811461441058,-98.18921995702834) circle (9.272673989617617);
\draw[filled] (49.39695136142308,-48.84479361199515) circle (9.725262526508178);
\draw[filled] (-238.9398053485659,-168.3095611650496) circle (10.157705425009988);
\draw[filled] (-76.92668110251373,-148.23691889127838) circle (10.572475007917987);
\draw[filled] (42.356150351585704,-138.66113332090075) circle (10.971575825050088);
\draw[filled] (-42.77151365254189,-10.011204842846626) circle (11.356659912870361);
\draw[filled] (77.63229714781437,-163.9973345042055) circle (11.729107922956874);
\draw[filled] (-225.54497488379042,-187.90991228484498) circle (12.090087715155041);
\draw[filled] (-110.66630387381855,-83.07113940950926) circle (12.440597624387506);
\draw[filled] (96.67388593363385,-87.73462893072943) circle (12.781499033512818);
\draw[filled] (-96.33224163995908,-60.17053793984311) circle (13.11354131558147);
\draw[filled] (-136.5443500885646,-162.50084391013695) circle (13.43738122281404);
\draw[filled] (93.03874621832907,-186.2464018373733) circle (13.753598162626698);
\draw[filled] (-27.635538403430658,-32.98925126047295) circle (14.062706379313392);
\draw[filled] (27.474085961974783,-114.88389352538267) circle (14.365164774639888);
\draw[filled] (-183.13570200445125,-149.63476509329982) circle (14.661384903696094);
\draw[filled] (-31.267355224763087,-185.0482624559518) circle (14.951737544048212);
\draw[filled] (15.218592530953416,8.984730902876759) circle (15.236558137514981);
\draw[filled] (-15.485725290496761,-156.12905810368346) circle (15.516151332390615);
\draw[filled] (138.58860950164305,-144.29438712272494) circle (15.790794801447461);
\draw[filled] (-265.6678436725432,-183.93925752795988) circle (16.060742472040115);
\draw[filled] (-16.295801233412814,-83.54864586436554) circle (16.326227275306643);
\draw[filled] (16.569497892621797,-66.59051972919656) circle (16.587463499183336);
\draw[filled] (-106.26225268534824,-163.9344465046336) circle (16.84464881285783);
\draw[filled] (47.7293193500269,-182.90203398294742) circle (17.097966017052606);
\draw[filled] (-140.45455342995024,-84.13652277850389) circle (17.347584564199675);
\draw[filled] (-72.58380455206606,-29.458295755351042) circle (17.593661884435317);
\draw[filled] (71.9406332402704,-118.3550732976393) circle (17.836344546896516);
\draw[filled] (-161.366468471518,-181.92423071934758) circle (18.075769280652434);
\draw[filled] (35.35330541158377,-20.775537315345602) circle (18.312063875466123);
\draw[filled] (-35.968050792140545,-126.03963924654587) circle (18.545347979235235);
\draw[filled] (-89.82780213049766,-118.8513364174661) circle (18.775733806237707);
\draw[filled] (164.16675377841358,-180.99667323191883) circle (19.003326768081177);
\draw[filled] (-217.0253140686143,-148.60152925704304) circle (19.22822603742331);
\draw[filled] (19.43255944645483,-158.7848113377401) circle (19.450525053016356);
\draw[filled] (-19.6398859314784,10.768789788585947) circle (19.670311973372208);
\draw[filled] (-67.87605704739737,-180.11232991470436) circle (19.88767008529564);
\draw[filled] (57.81398813162299,-80.41344979158733) circle (20.10267817265491);
\draw[filled] (-155.9416619300299,-127.6369169267644) circle (20.315410850019976);
\draw[filled] (102.27139594952835,-143.15933988444604) circle (20.525938865174535);
\draw[filled] (-55.76731530060816,-67.79132599824078) circle (20.734329373979172);%
        \end{tikzpicture}%
     \\
    \end{tabular}
    \caption{Example packings of various sets of circles into a right triangle produced by Split Packing.}
    \label{fig:example-circles-in-right}
\end{figure}
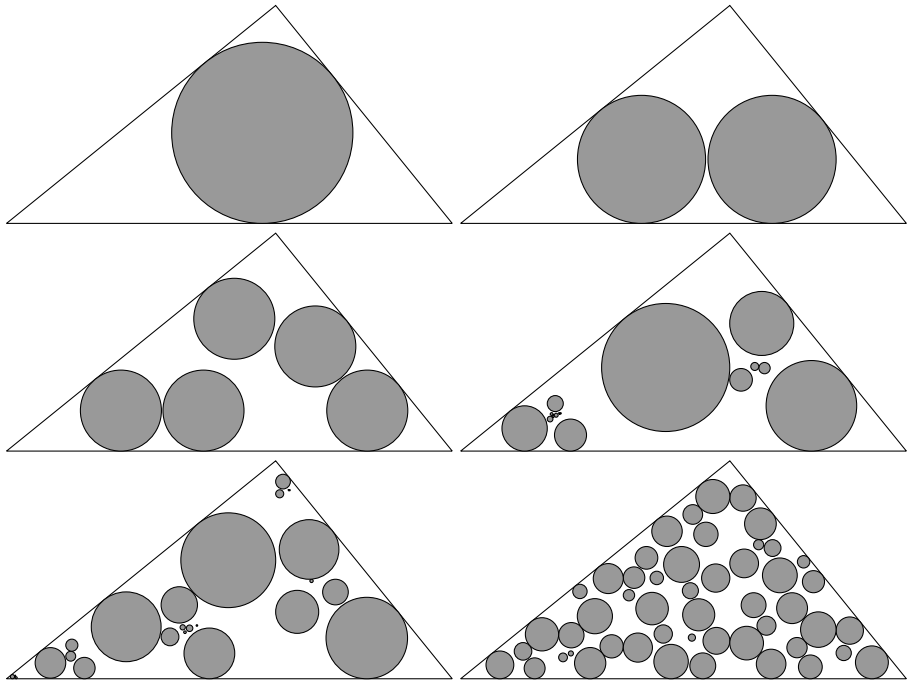

\section{The Problem with Acute Triangles}\label{sec:acute-triangles}

A class of triangles for which we have not succeeded in proving the critical density are acute triangles.
The problem is that the condition for \Cref{th:overlap} is not met, which means that the two hats may overlap.

        \begin{figure}%
            \begin{tikzpicture}[scale=2.5]%
                \input{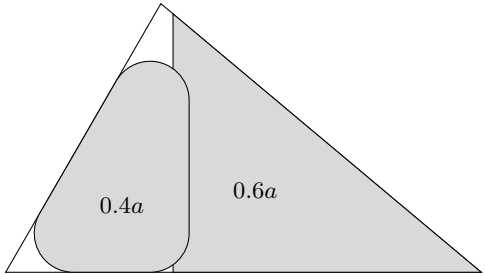}%
            \end{tikzpicture}%
            \caption{For acute triangles, the two hats may overlap.}%
            \label{fig:overlap}%
        \end{figure}%

The following term is useful for discussing worst cases:

\begin{definition}
    A shape's \emph{twincircles} are the largest two equal circles that can be packed into the shape.
\end{definition}

We work under the following assumption:

\begin{conjecture}
    A set of circles can be packed into a triangle if the circles' combined area does not exceed the triangles incircle or twincircle, whichever is smaller.
\end{conjecture}

If this conjecture is true, surely there are strategies which can pack into acute triangles with critical density.
For example, we attempted to split the set of circles into four subsets using a slightly modified \textsc{Split} algorithm, and then to pack those four hats into the container, like in \Cref{fig:four}.
Again, this is motivated by the observation that, when splitting each circle top-down into four equal circles, this strategy always works because the triangle is recursively divided into four similar triangles.

        \begin{figure}%
            \begin{tikzpicture}[scale=2.5]%
                \input{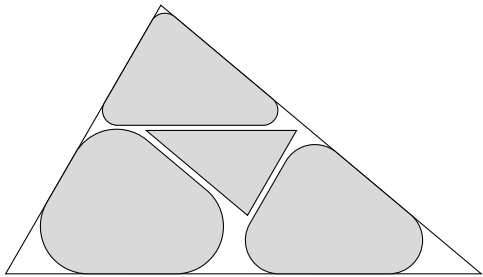}%
            \end{tikzpicture}%
            \caption{Packing four hats into an acute triangle.}%
            \label{fig:four}%
        \end{figure}%

Unfortunately, this strategy fails for some instances, as depicted in \Cref{fig:four-fail}.
For this instance, the largest group, consisting of a single circle, cannot be packed if \emph{any} of the smaller groups is packed into the top or the left corner, because the remaining free space is not wide enough for the circle.
Other strategies or a case distinction would need to be considered.

        \begin{figure}%
            \begin{tikzpicture}[scale=3]%
                \input{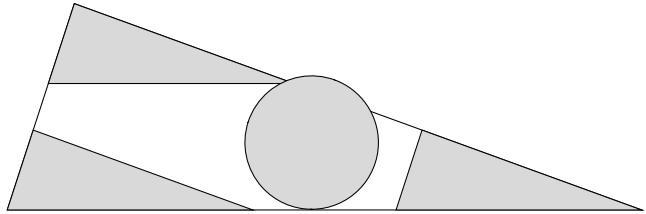}%
            \end{tikzpicture}%
            \caption{This strategy does not work for a container with a right-angle of $\frac{\pi}{10}$, incircle $1$ and the set of circles $\{0.55,0.15,0.15,0.15\}$.}%
            \label{fig:four-fail}%
        \end{figure}%

\section{Conclusion}\label{sec:conclusion}

In this paper, we presented a constructive proof of the critical densities when packing circles into squares, as well as right or obtuse triangles, using a weighted Split Packing technique.
We see more opportunities to apply this approach in the context of other packing and covering problems.

It is possible to use Split Packing to pack into \textbf{other container types}.
At this point, we can establish the critical densities for packing circles into equilateral triangles and rectangles exceeding a certain aspect ratio.
One could also consider the problem of packing into circles, ovals, regular polygons, or generalized quadrilaterals.
For some of these container types, even the worst-case instance does not seem obvious.
For circular and “almost square” rectangular containers, we assume the worst cases would again be their twincircles, see \Cref{fig:circle-rect-worst}, but it is unclear how to deal with the resulting shapes when cutting along the circles' tangent: Compared to triangular and square containers, these shapes cannot be split into self-similar pieces.
It is also possible that the depicted instances are not the actual worst cases.

        \begin{figure}%
            \begin{tikzpicture}[scale=1.2]%
                    \draw (0,0) circle (1.5);
    \draw[filled] (0.75,0) circle (0.75);
    \draw[filled] (-0.75,0) circle (0.75);
\end{tikzpicture}
\hspace{1.5cm}
\begin{tikzpicture}[scale=1.2]
    \draw (0,0) rectangle (4,3);
    \draw[filled] (1.05051,1.05051) circle (1.05051);
    \draw[filled] (4-1.05051,3-1.05051) circle (1.05051);%
            \end{tikzpicture}%
            \caption{Assumed worst-case instances for a circle and a near-square rectangle.}%
            \label{fig:circle-rect-worst}%
        \end{figure}
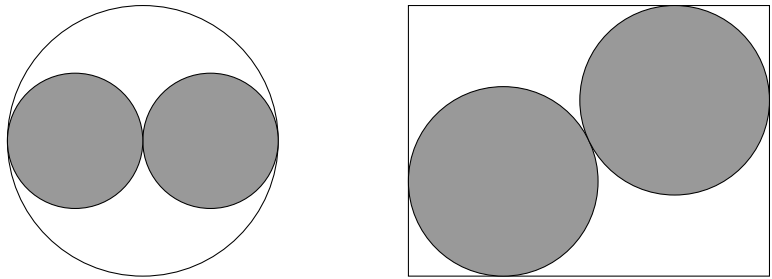%

Also, the problem of finding the critical density for packing into \textbf{acute triangles} is still open.
See \Cref{sec:acute-triangles} for a discussion on why the Split Packing approach does not directly work for acute ones.
A strategy for packing acute triangles with critical density, combined with the results of this thesis, would give an elegant, general result for all triangles.

Split Packing can also be extended to pack \textbf{objects other than circles}.
We can establish the critical densities for packing octagons into squares, and think we can describe the maximum shape which can be packed into squares using Split Packing.
Objects like ovals, rectangles, or even more general convex objects could be considered.
For these modified problems, again, it does not seem obvious what the worst-case packings would look like.

Another natural extension is the \textbf{online version} of the problem.
The current best algorithm that packs squares into a square in an online fashion by Brubach~\cite{brubach2014improved}, based on the work by Fekete and Hoffmann~\cite{FH2013online,FH2017online}, gives a density guarantee of $\frac{2}{5}$.
It is possible to directly use this algorithm to pack circles into a square in an online situation with a density of $\frac{\pi}{10} \approx 0.3142$.
It would be particularly interesting to see whether some form of online Split Packing would give better results.

Our original motivation stemmed from origami design.
When only packing circles, the resulting origami structures resemble arbitrary stars.
When one wants to design general tree-shaped structures, it is necessary to introduce separating pathways between the circles, a technique called \textbf{circle/river packing}, pioneered by Lang~\cite{lang1996computational}.
A packing scheme like Split Packing seems promising because it often introduces gaps inbetween two subgroups anyway.
At this point, we can establish a constant-factor approximation for perfectly symmetric binary trees (see \Cref{fig:origami}), but we do not know how to approximate the paper size needed for crease patterns of general trees.

        \begin{figure}%
            \begin{tikzpicture}[scale=6.5]%
                \draw (0,0) -- (0.1,0.1) -- ++(0,0.141) -- ++(-0.04,0.04);
\draw (0.1,0.241) -- ++(0.04,0.04);
\draw (0.1,0.1) -- ++(0.141,0) -- ++(0.04,0.04);
\draw (0.241,0.1) -- ++(0.04,-0.04);

\draw (-0,-0) -- (-0.1,-0.1) -- ++(-0,-0.141) -- ++(0.04,-0.04);
\draw (-0.1,-0.241) -- ++(-0.04,-0.04);
\draw (-0.1,-0.1) -- ++(-0.141,0) -- ++(-0.04,-0.04);
\draw (-0.241,-0.1) -- ++(-0.04,0.04);

\end{tikzpicture}
\hspace{1.5cm}
\begin{tikzpicture}[scale=3.8]
    \draw (0,0) rectangle (1,1);

    \draw (0.2,1) -- (1,0.2);
    \draw (1,0.4) -- (0.8,0.6) -- (1,0.8);
    \draw (0.4,1) -- (0.6,0.8) -- (0.8,1);
    \draw (0.6,0.6) -- (1,1);
    \draw[filled] (0.6-0.0585,1-0.0585) circle (0.0585);
    \draw[filled] (0.6+0.0585,1-0.0585) circle (0.0585);
    \draw[filled] (1-0.0585,0.6-0.0585) circle (0.0585);
    \draw[filled] (1-0.0585,0.6+0.0585) circle (0.0585);

    \draw (0.8,0) -- (0,0.8);
    \draw (0,0.6) -- (0.2,0.4) -- (0,0.2);
    \draw (0.6,0) -- (0.4,0.2) -- (0.2,0);
    \draw (0.4,0.4) -- (0,0);
    \draw[filled] (0.4-0.0585,0.0585) circle (0.0585);
    \draw[filled] (0.4+0.0585,0.0585) circle (0.0585);
    \draw[filled] (0.0585,0.4-0.0585) circle (0.0585);
    \draw[filled] (0.0585,0.4+0.0585) circle (0.0585);%
            \end{tikzpicture}%
            \caption{A folding of the tree on the left can be realized by a crease pattern based on the circle/river packing on the right.}%
            \label{fig:origami}%
        \end{figure}
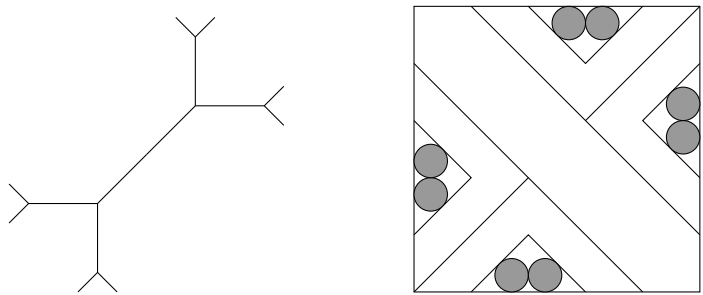%

It seems like a natural extension to apply Split Packing to \textbf{three-di\-men\-si\-o\-nal packing} problems.
For example, one could try to pack spheres into a cube using a Split Packing approach.
Unfortunately, this does not directly seem to work out: Assuming the worst case are two equally sized spheres packed into opposite corners of the cube, one would like to be able to cut the cube along the spheres' tangential plane.
This results in two shapes as depicted on the right in \Cref{fig:cube-half}, but it is not possible to fit two quarter-spheres into each of these polyhedra.
Still, any extensions regarding three dimensional problems would be notable.

\begin{figure}
    \includegraphics[width=0.4\textwidth]{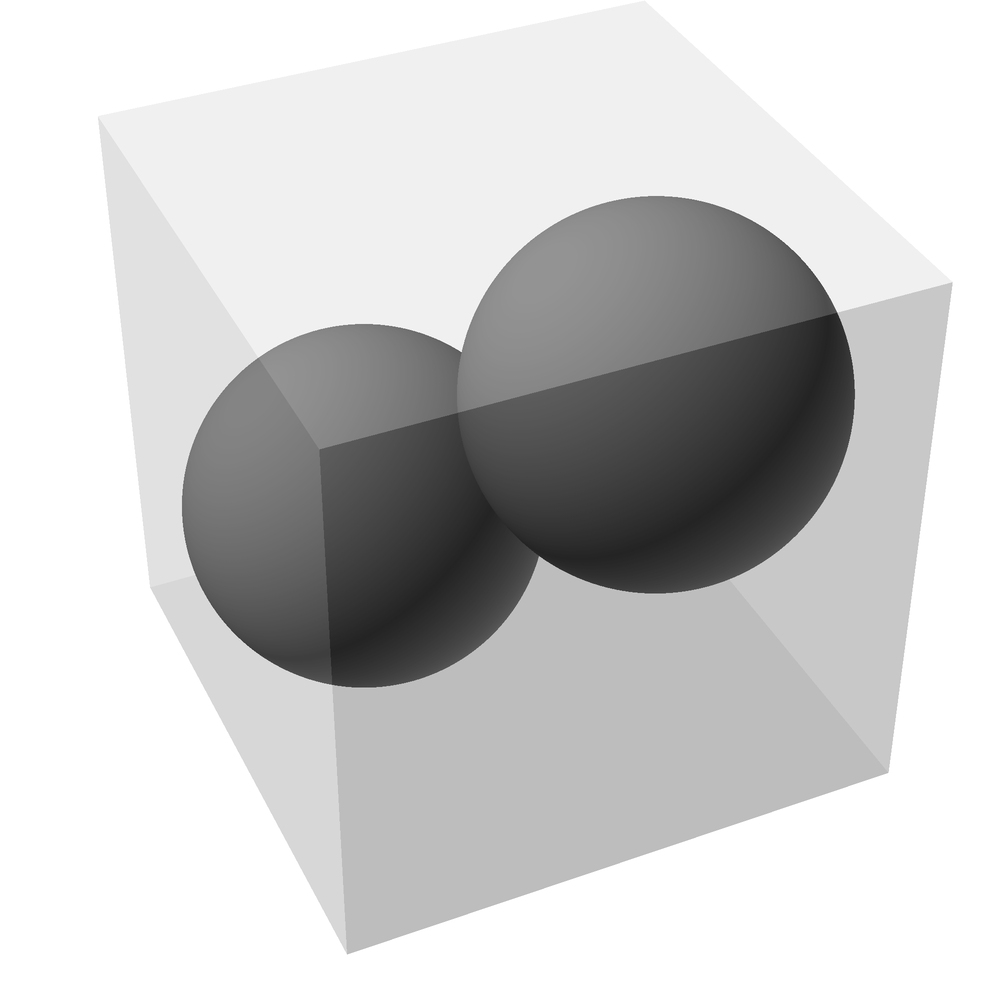}
    \hspace{1cm}
    \includegraphics[width=0.4\textwidth]{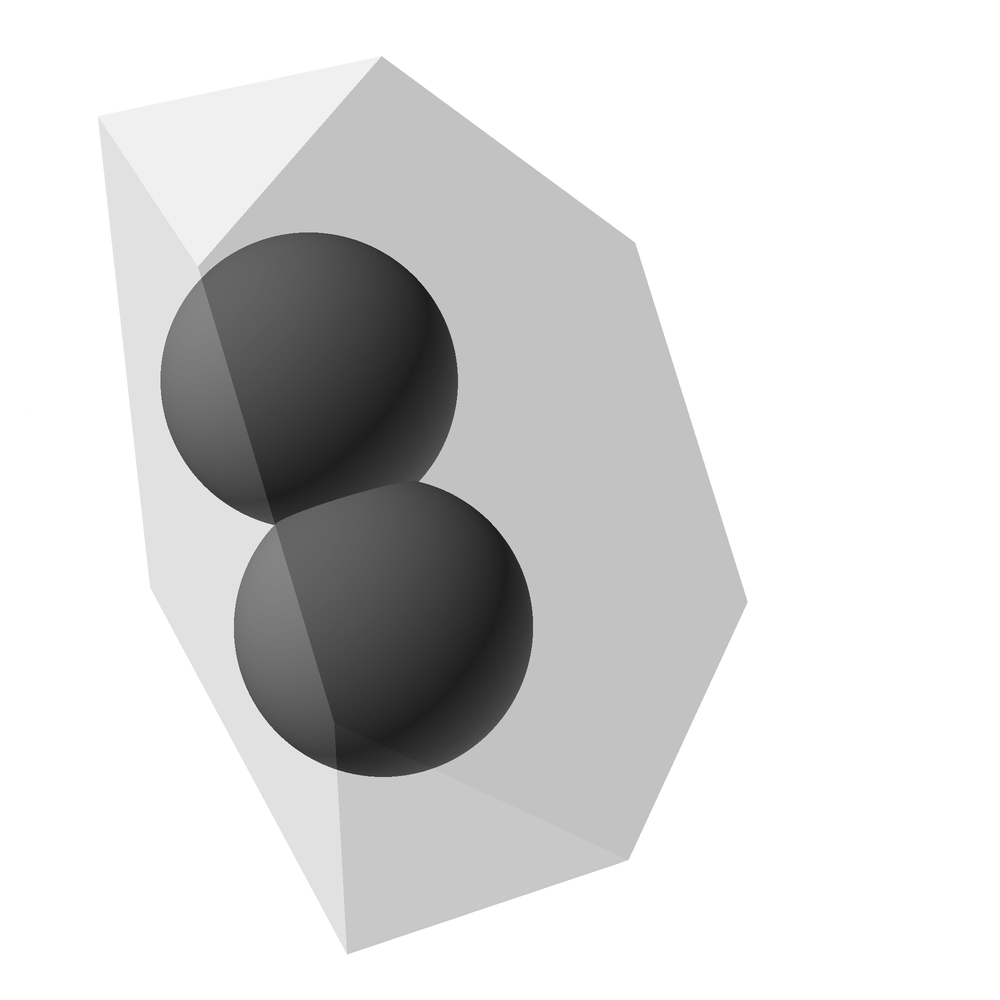}
    \caption{Left: Assumed worst case for packing spheres into a cube. Right: Two quarter-spheres do not fit in a half.}\label{fig:cube-half}
\end{figure}

Instead of packing circles into containers, one could ask a question which is in some sense the opposite problem: What is the smallest area so that we can always \textbf{cover the container} with circles of that combined area? For example, if we want to cover an isosceles right triangle, and restrict ourselves to at most two circles in our input set, the area of a circle whose diameter equals the triangle's hypotenuse is sufficient, see \Cref{fig:cover}.
To generalize this method, it would now be sufficient to show that all sets of circles with a combined area equal to the area of the left circle can cover the quadrilateral on the left, but it does not seem trivial to find an argument for that.

        \begin{figure}%
            \begin{tikzpicture}[scale=1.1]%
                    \draw[draw=none] (0,-1) -- (0,1.2); 
    \draw (-1,0) -- (1,0) -- (0,1) -- cycle;
    \draw[cover] (0,0) circle (1);
\end{tikzpicture}
~
\begin{tikzpicture}[scale=1.1]
    \draw[draw=none] (0,-1) -- (0,1.2); 
    \draw (-1,0) -- (1,0) -- (0,1) -- cycle;
    \draw (0.6,0.4) -- (0.6,0);
    \draw[cover] (-0.15,0.25) circle (0.9591);
    \draw[cover] (0.8,0.2) circle (0.2828);
\end{tikzpicture}
~
\begin{tikzpicture}[scale=1.1]
    \draw[draw=none] (0,-1) -- (0,1.2); 
    \draw (-1,0) -- (1,0) -- (0,1) -- cycle;
    \draw (0.2,0.8) -- (0.2,0);
    \draw[cover] (-0.4,0.4) circle (0.824);
    \draw[cover] (0.6,0.4) circle (0.565);
\end{tikzpicture}
~
\begin{tikzpicture}[scale=1.1]
    \draw[draw=none] (0,-1) -- (0,1.2); 
    \draw (-1,0) -- (1,0) -- (0,1) -- cycle;
    \draw (0,1) -- (0,0);
    \draw[cover] (-0.5,0.5) circle (0.707);
    \draw[cover] (0.5,0.5) circle (0.707);%
            \end{tikzpicture}%
            \caption{An isosceles right triangle can always be covered by two circles with a combined area of its excircle.}%
            \label{fig:cover}%
        \end{figure}
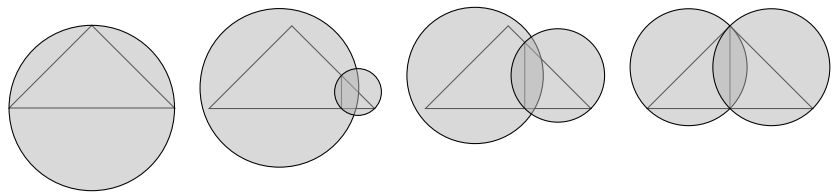%

\section*{Acknowledgements}

We thank Erik Demaine, Dominik Krupke, Christian Neukirchen, and Jan-Marc Reinhardt for useful discussions and support.
We also thank several anonymous reviewers for their helpful comments.

\bibliographystyle{acm}
\bibliography{references}

\end{document}